\documentclass[journal]{IEEEtran}

\usepackage{graphicx}
\usepackage{mdwlist}
\usepackage{amsmath}
\usepackage{amsthm}
\usepackage{amssymb}
\usepackage{url}
\usepackage{xspace}
\usepackage{balance}
\usepackage{algorithm}
\usepackage{algorithmicx}
\usepackage{algpseudocode}
\usepackage{multirow}
\usepackage{subfigure}
\usepackage{setspace}
\usepackage{times}
\usepackage{float}
\usepackage{cite}
\usepackage{epstopdf}
\usepackage{ifthen}

\def\unit#1#2{\hbox{#1$\,\textrm{#2}$}}
\newtheorem{thm}{\textbf{Theorem}}
\newtheorem{defn}{\textbf{Definition}}

\newtheorem{prop}{\textbf{Proposition}}

\newtheorem{prob}{\textbf{Problem}}

\theoremstyle{definition}
\newtheorem{example}{\textbf{Example}}

\newcommand{\comment}[1]{}
\newcommand{\version}{long}  
\newcommand{\longv}{\ifthenelse{\equal{\version}{long}}}

\newcommand{\figcaption}[1]{\vspace*{-3mm}\caption{#1}\vspace*{-3mm}}

\title{Trajectory Aware Macro-cell Planning \\for Mobile Users}

\author{

    \IEEEauthorblockN{Shubhadip Mitra\IEEEauthorrefmark{2}, Sayan Ranu\IEEEauthorrefmark{3}, Vinay Kolar\IEEEauthorrefmark{1}, Aditya Telang\IEEEauthorrefmark{1}, \\Arnab Bhattacharya\IEEEauthorrefmark{2}, Ravi Kokku\IEEEauthorrefmark{1}, Sriram Raghavan\IEEEauthorrefmark{1}}\\
    \IEEEauthorblockA{\IEEEauthorrefmark{1}IBM Research, India.
    \\\{vinkolar, aaditya.telang, ravkokku, sriramraghavan\}@in.ibm.com}\\
    \IEEEauthorblockA{\IEEEauthorrefmark{2}Dept. of Computer Science and
	Engineering, Indian Institute of Technology, Kanpur, India.
    \\\{smitr, arnabb\}@cse.iitk.ac.in}\\
		\IEEEauthorblockA{\IEEEauthorrefmark{3}Dept. of Computer Science and
	Engineering, Indian Institute of Technology, Madras, India.
    sayan@cse.iitm.ac.in}
}

\newcommand{\tump}{\ensuremath{\operatorname{TUMP}}\xspace}
\newcommand{\tumpo}{\ensuremath{\operatorname{TUMP}(1)}\xspace}
\newcommand{\tumpg}{\ensuremath{\operatorname{TUMP}(\gamma)}\xspace}
\newcommand{\tumpe}{\ensuremath{\operatorname{TUMP}(\epsilon)}\xspace}
\newcommand{\dmax}{\ensuremath{d}\xspace}

\newcommand{\todo}[1]{}

\newcommand{\randk}{RANDOMIZED\xspace}
\newcommand{\simg}{SIMPLE-GREEDY\xspace}
\newcommand{\incg}{INC-GREEDY\xspace}
\newcommand{\decg}{DEC-GREEDY\xspace}
\newcommand{\lpr}{LP\xspace}

\begin{document}

\maketitle

\begin{abstract} 
In this paper, we handle the problem of efficient user-mobility driven macro-cell planning in cellular networks. As cellular networks embrace heterogeneous technologies (including long range 3G/4G and short range WiFi, Femto-cells, etc.), most traffic generated by static users gets absorbed by the short-range technologies, thereby increasingly leaving mobile user traffic to macro-cells. To this end, we consider a novel approach that factors in the trajectories of mobile users as well as the impact of city geographies and their associated road networks for macro-cell planning. Given a budget $k$ of base-stations that can be upgraded, our approach selects a deployment that improves the most number of user trajectories. The generic formulation incorporates the notion of quality of service of a user trajectory as a parameter to allow different application-specific requirements, and operator choices. We show that the proposed trajectory utility maximization problem is NP-hard, and design multiple heuristics. To demonstrate their efficacy, we evaluate our algorithms with real and synthetic datasets emulating different city geographies. For instance, with an upgrade budget k of 20\%, our algorithms perform 3-8 times better in improving the user quality of service on trajectories when compared to greedy location-based base-station upgrades.
\end{abstract}

\section{Introduction}
\label{sec:Intro}

As cellular networks advance towards providing high-bandwidth services to a 
large subscriber base, the revenue growth for cellular network operators is 
significantly slowing down~\cite{booz}.
On one hand, operators are making heavy investments to upgrade the network to
cater to the growing bandwidth demands. On the other hand, the revenue per byte is 
decreasing because of increased competition and demand for cheaper services. 
Consequently, to keep network deployment costs low, cellular operators are 
increasingly focusing on short-range technologies, such as Small-cells and 
Femto-cells, for meeting the high-bandwidth demands from customers~\cite{smallcell,femtocell}. 

While short-range technologies provide high bandwidth to static users at a much 
cheaper cost per byte, long-range networks are more appropriate to provide continuous 
coverage. Hence, as short-range networks absorb static user traffic, macro-cells will 
increasingly handle mobile user traffic. As a result, it becomes important to 
consider {\em mobility patterns of users} for effective macro-cell upgrades. 
For example, users are accessing a variety of applications such as YouTube and 
maps when mobile. These applications require high bandwidth and delay guarantees to 
ensure high quality of experience (QoE). One recent survey indicates that data traffic 
increases by 20-30\% during busy commute hours~\cite{wandera}. Hence, it is imperative 
for the operators to plan macro-cell upgrades to {\em maximize the quality of experience} 
for mobile subscribers.

To the best of our knowledge, no current approach, either in research 
literature or in practice, considers user mobility trajectories and 
the experience perceived by users for macro-cell upgrades. 
Operators currently perform incremental upgrade of a few 
base-stations to newer technology while installing cheaper older technology 
base-station in areas with low demand. For example, a major 
operator, Airtel in India, deployed 6,728 new 3G sites (27\% year-on-year growth), 
while also deploying 4,977 new 2G sites in 2013-2014~\cite{airtel}. 

Operators deploy higher generations of technology based on the anticipated 
static user population. This often leads to switching between base-stations of 
multiple generations of technologies on a mobile user trajectory, thereby 
leading to degraded quality of experience for mobile users. For instance, a 
mobile 4G user on a given daily commute trajectory may often handoff from  
 a 4G cell to a 3G/2G cell, depending on the current deployment. 

In this paper, we focus on explicitly incorporating the experience 
of mobile users on their trajectories for incremental upgrade of cellular 
networks from one generation to another. Note that upgrades often happen at 
cell-towers that already have a previous generation of the technology deployed, 
mainly to keep real-estate costs for cell-sites (including land, room, grid 
connection, diesel generator, backhaul provisioning, etc.) low. Hence, we 
focus on the problem of identifying base-stations that need to change from one 
generation to another, and leave RF-level planning~\cite{Song2010,atoll} as a 
follow-up task that field engineers perform at the towers identified for upgrades.

Specifically, we tackle the following budget-constrained Trajectory Utility 
Maximization Problem (\tump): {\em Given a macro-cell 
network with $n$ base-stations and quality of experience of mobile users 
when connected to a sequence of base-stations, how do we choose $k$ 
base-stations such that the overall mobile user experience is maximized?} 

The key idea of the \tump problem is to: (1) consider the performance perceived 
by users on their trajectories using call/transaction records, 
(2) identify the base-stations that provide lower QoS on each trajectory, 
and (3) deploy higher capacity base-stations such that maximum number 
of trajectories are bottleneck-free. 


\noindent We make the following contributions:
\begin{enumerate*}
\item To the best of our knowledge, this is the first paper to propose 
cellular macro-network planning by considering users\rq{} trajectories. 
We develop an extensible framework (\tump) for optimizing mobile user experience 
for different trajectory utility functions that capture different application QoS 
requirements.

\item We show that \tump is NP-hard. We then present two heuristics:
\incg and \decg. We prove bounds on their efficacy, and show 
that they can be incrementally applied to an evolving network; i.e., as and when the operator allocates 
additional budget, these 
algorithms can be applied to incrementally evolve the network from one	generation to 
another. Our techniques enable operators to satisfy 3-8 times more number of
mobile users than an approach that uses a greedy location-based base-station 
upgrade.

\item We measure and analyze the existing experiences of mobile users during their 
everyday commute. We show that around 50\% of the base-stations during commute 
provide throughputs that cannot cater to a medium sized video. Users often 
experience long stretches of time that degrade user experience.
 
\item We develop a network trace generator from how people move in large cities. 
We believe that the trace generator is useful for research beyond 
this paper since cellular network traces are most often not published openly by the 
operators. Through these traces, we show that the investment required to provide 
satisfactory QoS to mobile users is dependent on the population distributions 
and their road-network. Specifically, cities with a dense central business 
districts, such as New York, need less budget to satisfy a large segment of
mobile users than cities where businesses are spread out (e.g., Atlanta).
\end{enumerate*}



The rest of the paper is organized as follows. We describe a measurement-based analysis 
of the problem in Section~\ref{sec:measurement}. Section~\ref{sec:formulation} 
discusses the proposed trajectory utility maximization problem (\tump). Section~\ref{sec:algo} 
analyzes the approximation algorithms and their bounds. In Section~\ref{sec:methodology} 
and Section~\ref{sec:results}, we evaluate the efficacy of the algorithms under different 
parameter settings. Section~\ref{sec:related} describes related work while 
Section~\ref{sec:concl} discusses avenues for future work and concludes.

\begin{table}[t]
\caption{Route Characteristics} 
\centering 
{\scriptsize
\begin{tabular}{| r | c | c | c | c|}
\hline
Route & Drive Time & Num Trajectories & Num BS & Approx start hour\\
\hline 
RT-1 & 26 hours & 13 & 100 & 8:00 AM\\ 
RT-2 & 25 hours & 18 & 76 & 6:15 PM\\ 
RT-3 & 13 hours & 7 & 56 & 7:45 PM\\ 
\hline
TOTAL & 64 hours & 38 & 158 &\\
\hline 
\end{tabular}} 
\label{tab:routes} 
\end{table} 

\section{A Measurement-based Motivation}
\label{sec:measurement}
We conducted measurement-based experiments to quantify user 
experiences during mobility. We developed an Android app that measures the throughput 
while users are traveling on their daily trajectories. The users have a 3G data 
connection. We demonstrate that the users often experience prolonged periods of 
non-satisfactory download rates along their trajectories. 

Our measurements are aware of the data limits imposed on the cellular plan. 
Hence, our app was tuned to: (1) sample throughput with a 
high-periodicity of approximately five minutes (only when the user is traveling); and,
(2) each \textit{sample} consists of downloading chunks of \unit{50}{kB} for 5 times 
consecutively; the size of each sample was chosen such that there are at least 30 TCP 
packets (to accommodate TCP's startup delays). For each sample, we record the 
base-station id, technology of the base-station, and GPS coordinates of the user. 

We collected over \unit{20} {days} of driving data on three main routes. 
Each route is around \unit{25} {Km}, and contains multiple trajectory samples over 
different days. Table~\ref{tab:routes} summarizes the different routes that were
taken by the users. 

\begin{figure}[t]
\begin{center}
\includegraphics[width=2.5in]{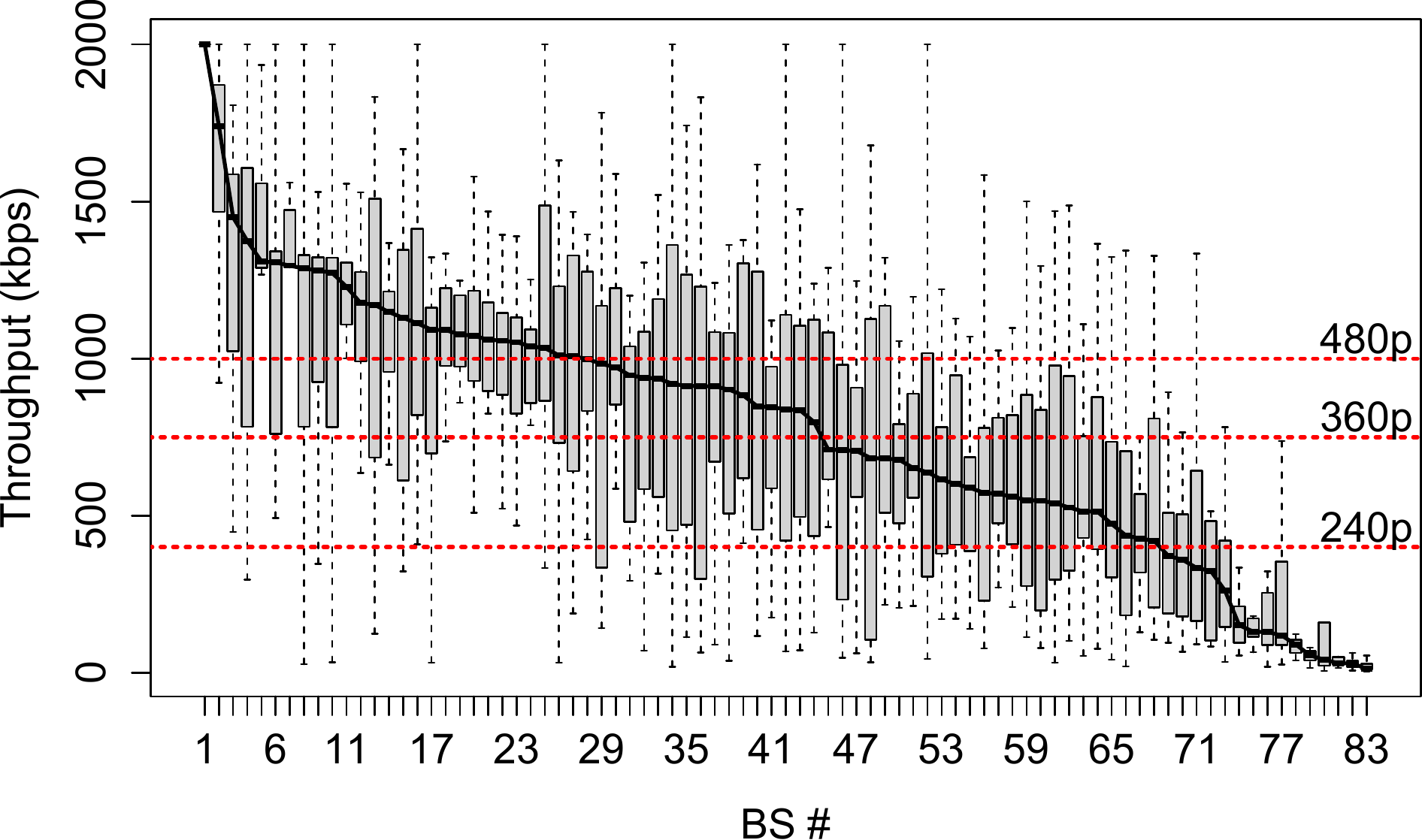}
\end{center}
\caption{Throughput distribution on base-stations}
\label{fig:meas_bs}
\end{figure}
We observed $158$ unique base-stations, out of which $83$ had 
$20$ or more samples. Fig.~\ref{fig:meas_bs} shows 
the throughput statistics on these $83$ base-stations. All the base-stations, 
except the last two, were indicated as 3G base-stations; the last two were 2G. 
While 3G currently claims to provide \unit{2}{Mbps} in the measured 
regions, we observed that most base-stations provide much lower throughput 
during every day mobility scenario. Such achievable throughput cannot cater to 
the demanding applications, such as video streaming, that mobile users often 
desire to use during everyday long commutes~\cite{wandera}.

For example, around 50\% of base-stations that we sampled cannot cater to the 
recommended bit-rate for a medium sized YouTube stream (320p video), which 
requires \unit{750}{kbps}. Recommended bit-rates for YouTube videos of 
other sizes are shown by the dotted red line in Fig.~\ref{fig:meas_bs}. 

We now analyze the variations of throughputs across base-stations on user trajectories. 
Fig.~\ref{fig:meas_time_th} demonstrates the fluctuations in throughputs as the user is traveling 
on RT-2 on various days. Each row plots one trajectory, with 
the dots representing the median throughput achieved. The radius of the 
dot is proportional to the throughput, and the dots are color-coded at 4 thresholds 
corresponding to different YouTube size videos. 
Fig.~\ref{fig:meas_time_th}
shows that mobile users often receive low capacity for extended durations of time. 
On RT-2, more than 47\% of the throughput sampled on the trajectories are below 
\unit{400 kbps}, which is the recommended bit rate for a lower quality (240p) 
YouTube video.

\begin{figure}
\begin{center}
\subfigure[Throughputs on different trajectories in RT-2
~\label{fig:meas_time_th}]{\includegraphics[width=6cm]
{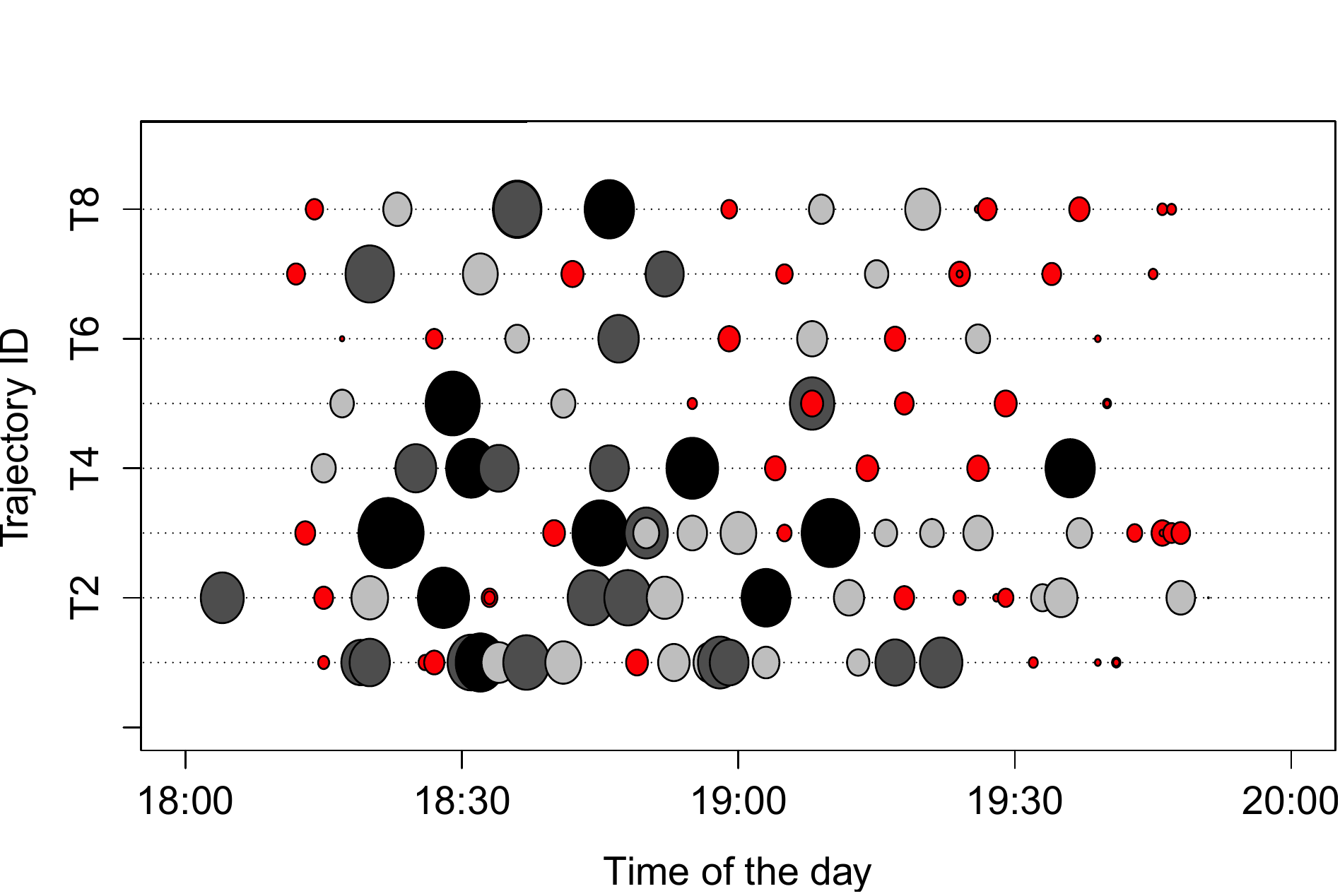}}
\subfigure[Base-station(BS) graph
~\label{fig:meas_graph}]{\includegraphics[width=\linewidth]
{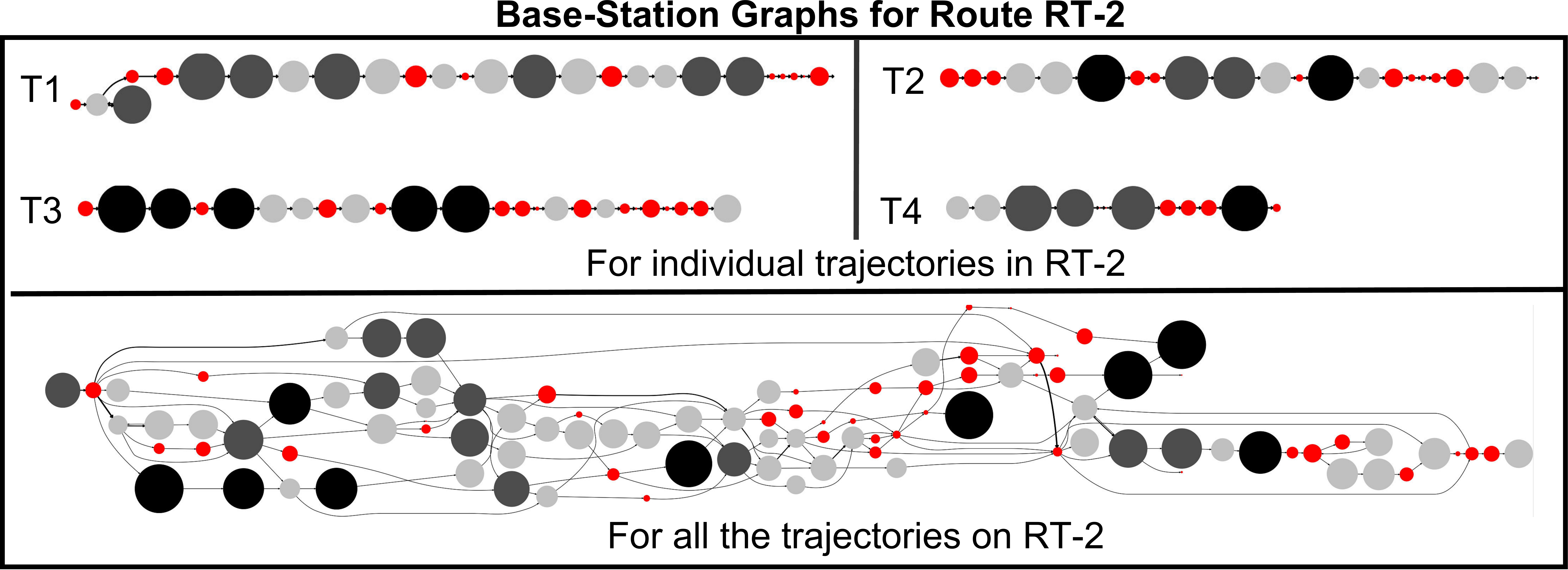}}
\end{center}
\figcaption{Mobile users constantly experience low throughput on trajectories}
\label{fig:measure}
\end{figure}
We construct a \textit{base station graph (BS-graph)} for each trajectory, and demonstrate 
the variance of throughput. Each node in a BS-graph is a base-station. Ideally, a 
directed edge is drawn between two base-stations if the user hands-off from one to the 
other. However, due to our limited sampling, an edge denotes that the user 
was connected to the two base-stations in consecutive samples. The radius of the 
node is proportional to the median throughput observed when connected to this 
base-station. 

Fig.~\ref{fig:meas_graph} shows the sequence of unique base-stations that the user 
had recorded on the trajectories. The lower half of Fig~\ref{fig:meas_graph} shows 
the BS-graph for all trajectories on RT-2, and the upper part of the figure separates
the BS-graph for individual trajectories. We use this notion of a chain of 
base-stations, with the user having some experience on each base-station, to 
theoretically abstract the mobile user experience on a trajectory. 
Fig.~\ref{fig:meas_time_th} and Fig.~\ref{fig:meas_graph} show that the user 
has {\em long stretches of time (and sequence base-stations) that degrade 
user experience}.


\section{Trajectory Utility Maximization Problem}
\label{sec:formulation}

\comment{

\begin{figure}[t]
\begin{center}
\includegraphics[width=\linewidth]{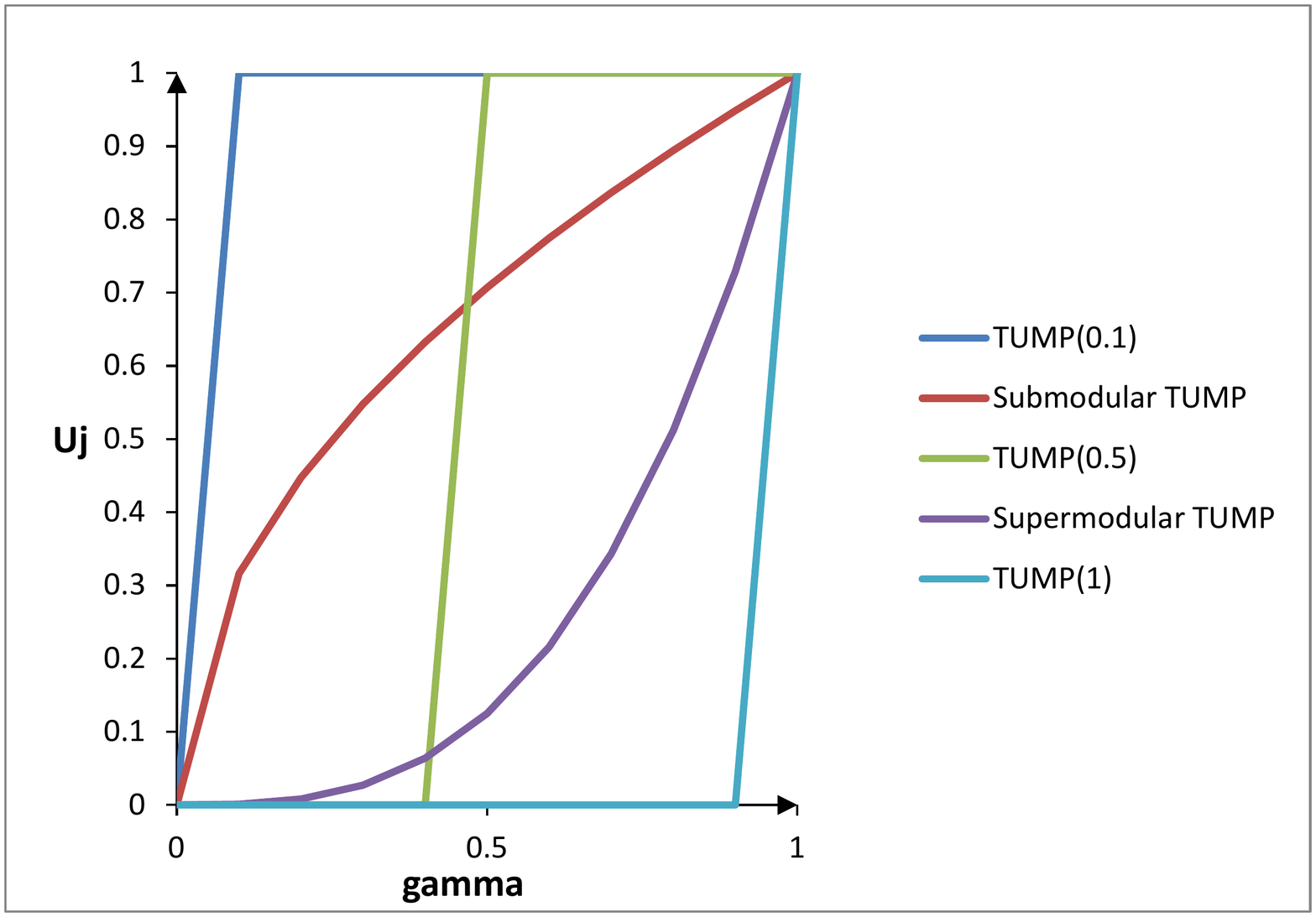}
\end{center}
\caption{Examples of variants of $TUMP$}
\label{fig:tump_variants}
\end{figure}

}

In this section, we design a generic formulation that, given a budget, plans
base-station upgrades that maximizes the mobile users\rq{} experience. Our model
is \emph{practical} and \emph{useful} for today's cellular operators as it
utilizes the already existing data and enables the operator to control key
parameters for base-station upgrades. Specifically, we allow the operator to:
(1) optimize based on a given budget; (2) define the strictness of when a
subscriber's mobile experience is considered to be poor; (3) utilize active data
already stored by many operators, such as call records~\cite{tdrs} and
deep-packet inspection logs~\cite{dpi}, to quantify user experience on a
trajectory. 


Consider the network $\mathcal{B}$ = \{$B_{1}$, \dots, $B_{n}$\} of $n$
base-stations spread across a region. 
A \emph{trajectory} $T_j$ is represented as a sequence of tuples of the form 
$\Phi_i = \langle B_{i}, \Delta_{i}, \eta_{i}\rangle$ that
captures 
the user experience. The user on this trajectory was connected to the 
base-station $B_{i} \in \mathcal{B}$ for a time interval of $\Delta_{i}$ units 
and received a throughput of $\eta_{i}$ bytes per unit of time. Note that 
$\eta_{i}$ can be any metric as long as a greater $\eta_{i}$ denotes better 
experience (e.g., throughput or packet success rate) when associated to a 
respective base-station. Henceforth, for brevity, we refer to this metric
as throughput. As we show in Section~\ref{sec:methodology}, the trajectories and
$\Phi_i$'s
can be constructed by scanning the active transaction records maintained by 
the operator.

For ease of notation, we write $B_i \in T_j$ if the  base-station
$B_i \in \mathcal{B}$ lies on the trajectory $T_j$. The length of a trajectory
$T_j$, denoted by
$|T_j|$, is simply the count of base-stations that lie on it. 
Suppose $\dmax$ denotes the maximum length of a trajectory.
\comment{
The maximum  length of any trajectory in $\mathcal{T}$ is denoted 
by $\dmax$.
}

For a trajectory $T_j$, a base-station $B_i \in T_j$ is a \emph{bottleneck
base-station} if it offers a degraded quality of service, e.g., an extremely low
upload/download speed, a call-drop, etc.  In our model, we assume that a
base-station $B_{i}$ acts as a bottleneck w.r.t. a trajectory $T_j$ if the
corresponding throughput is less than a threshold, i.e., $\eta_{i} < \tau$.  The
value of $\tau$ is computed from a
combination of network parameters.  A base-station that
is a bottleneck for one trajectory may \emph{not} be a bottleneck for other
trajectories since different users may experience different throughputs based on
various factors such as data plan, time of the day, etc.   
\comment{  
A trajectory $T_j$ is a \emph{bottleneck-free trajectory} if
there are no bottleneck base-stations in it.

The general problem that we address in the paper is as follows.
\begin{prob}
	Given a base-station network $\mathcal{B}$ of size $n$, a budget parameter
	$k$, and a set of $m$ trajectories $\mathcal{T} = \{T_1, \dots, T_m\}$,
	determine the set of $k$ base-stations to upgrade such that the number of
	trajectories that become bottleneck-free is maximized.
\end{prob}
\subsection{Trajectory Utility Maximization Problem}

We associate a utility function to each trajectory in terms of the total time covered by base-stations that are bottleneck-free, and then adopt the problem statement to
maximizing the sum of utilities of the trajectories.
}

Our goal is to maximally improve the mobile user experience by selectively
upgrading $k$ out of $n$ base-stations that act as bottlenecks on some
trajectories. 
Suppose $\mathbf{X}=\{x_1, \dots, x_n\}$ denotes the boolean solution vector such that $x_i
= 1$ if and only if base-station $B_i$ is chosen for upgradation and $0$
otherwise.

Our proposed framework allows the network operator to specify \emph{any trajectory utility function} $W_j:(T_j,\mathbf{X}) \rightarrow [0,1]$, defined over each trajectory $T_j$ and solution $\mathbf{X}$, that  captures the impact of base-station upgrades on the trajectory $T_j$.
 We assume that $W_j$ increases as more number of bottleneck base-stations on
 the trajectory $T_j$ get upgraded. In other words, the trajectory utility
 increases with its quality of experience (QoE).
 \comment{ 
 a trajectory experiencing satisfactory QoE would have high utility; and that witnessing unsatisfactory QoE would have low utility.
 }
  Our aim is to maximize the number of trajectories with high utility.
  
To do so, given any trajectory utility function $W_j$,  we map it to a step
utility function $U_j$ using a threshold $\gamma \ (0 \le \gamma \le1)$,
henceforth referred to as the \emph{bottleneck parameter}:
\begin{align}
	U_j =
			\begin{cases}
				1 & \text{ if } W_j \ge \gamma  \\
				0 & \text{ otherwise }
			\end{cases}
\label{eq:utility-def}	
\end{align}
We now formally state the Trajectory Utility Maximization Problem, \tumpg.

\begin{prob}[\tumpg]
	Given a base-station network $\mathcal{B}$ of size $n$, a budget parameter
	$k$, a bottleneck parameter $\gamma$, and a set of $m$ trajectories $\mathcal{T} = \{T_1, \dots, T_m\}$, each of which has an associated utility function $W_j$,
	determine the set of $k$ base-stations to upgrade such that the sum of
	utilities $\sum_{T_j \in \mathcal{T}} U_j$ is maximized, where $U_j$ is
	given by Eq.~\eqref{eq:utility-def}.
\end{prob}

Intuitively, a solution to a given instance of \tumpg with a high value of
$\gamma$ would benefit lesser number of \emph{spatially distinct}\footnote{Two
or more trajectories are \emph{spatially distinct} if their corresponding set of
base-stations are ``largely'' different.} trajectories, as it attempts to
utilize the available resources (i.e., upgraded base-stations) to optimize the
QoE on these trajectories. On the other hand, a solution to the same instance of
\tumpg with a lower value of $\gamma$ would attempt to be more fair in
distribution of the available resources across a larger set of spatially
distinct trajectories at the cost of allowing limited improvement in QoE. The
operator can judiciously tune $\gamma$ based on the budget and desired
subscriber experience.

\subsection{Solution Overview}
\label{sec:sol}

\begin{figure}[t]
\begin{center}
\includegraphics[width=0.75\linewidth]{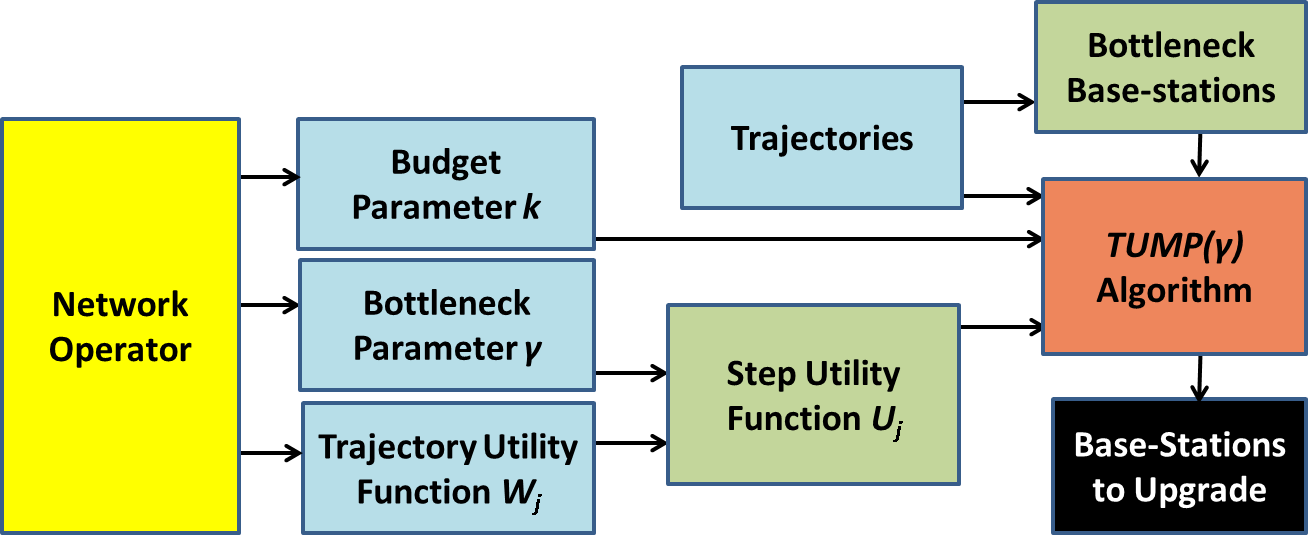}
\end{center}
\figcaption{Solution Overview}
\label{fig:sol}
\end{figure}

Fig.~\ref{fig:sol} outlines the basic steps involved in our solution framework.
We construct the user trajectories from call records that the operator has
already stored~\cite{tdrs,dpi}. Given a set of trajectories, we first identify
the bottleneck base-stations, based on the throughputs received.
A base-station is a candidate for upgrade if it is a bottleneck for any
trajectory. The network operator provides a
trajectory utility function $W_j$ associated with each trajectory $T_j$, and the
bottleneck parameter $\gamma$. To maximize the number of satisfied trajectories, 
we map the utility function $W_j$ to a step utility function $U_j$, using the
bottleneck parameter $\gamma$. Finally, we apply any of the proposed algorithms for
\tumpg (described in Section~\ref{sec:algo}), and report the $k$ base-stations to be upgraded.

\subsection{Bottleneck Utility Function}

Though the proposed framework allows any trajectory utility function, for the purpose of analysis and evaluation of our approach, this section introduces a special trajectory utility function, namely 
 \emph{bottleneck utility function}.

Given a trajectory $T_j$, we define the weight $w_{ji}$ for each base-station $B_{i} \in
T_j$, that accounts for the fraction of the total time that the user (on this trajectory)
was connected to the base-station $B_{i}$. More precisely, 
$w_{ji} =\frac{\Delta_{i}}{\sum_{B_{i}\in T_j} \Delta_{i}}$.
Suppose $b_{ji}$ denotes a bottleneck indicator variable that takes value $1$ if
the base-station $B_i \in T_j$ is a bottleneck base-station w.r.t. the
trajectory $T_j$, and $0$ otherwise.
%

Given a trajectory $T_j$ and solution $\mathbf{X}$, we define the \emph{bottleneck utility} function $W_j$ as follows:
\begin{align}
\label{eq:traj_weight}
	W_j = \sum_{B_i \in T_j, b_{ji} = 0} w_{ji}
		+ \sum_{B_i \in T_j, b_{ji} = 1} w_{ji} . x_i
\end{align}
 $W_j$ essentially captures the fraction of the total time when the
user enjoys acceptable QoE on the trajectory $T_j$.  If all the
base-stations on $T_j$ are non-bottleneck,
then $W_j = 1$; otherwise, $W_j < 1$.  Henceforth, we consider the bottleneck utility function as the default trajectory utility function.

Based on this, we next define a class of trajectories that enjoy satisfactory QoE after upgradation of the base-stations.

\begin{defn}[$\gamma$-bottleneck-free trajectory]  
	A trajectory $T_j \in \mathcal{T}$ is \emph{$\gamma$-bottleneck-free} if its
	utility $W_j \ge \gamma$ where $W_j$ is given by Eq.~\eqref{eq:traj_weight}, and 
	$\gamma \in [0,1]$ is  the
	bottleneck parameter.
\end{defn}

Of particular analytical interest is the problem instance $\tump(\gamma = 1)$, denoted henceforth by simply $\tumpo$. In this case, $\forall j=1,\dots,m,\ U_j
		=W_j= 1$ if and only if all the bottleneck base-stations on the trajectory
		$T_j$ are upgraded.  This problem instance is interesting because it is the worst case instance of the \tumpg problem that
		aims to maximize the number of trajectories that are \emph{completely} 	bottleneck-free. 

This framework allows the network operator to suitably select the
bottleneck parameter $\gamma$  based on
the application requirements. For example, $\gamma = 1$ is suitable
for real-time applications such as voice calls or video conferences whereas
$\gamma = 0.8$ may suffice for video streaming since video players can mask-off
certain durations of low connectivity by buffering.  Similarly, even $\gamma =
0.5$ may be enough for elastic applications such as background synchronization
of emails.

\longv
{
\subsection{Hardness of \tumpg}
}
{}
\longv
{
Next, we show  that \tumpg is  NP-hard due to a reduction from the $k$-Vertex
Cover ($k$-VC) problem \cite{Bourgeois}.
\begin{prob}[$k$-VC]
	Given an undirected graph $G(V,E)$ where $V$ is the set of $n$ nodes and $E$
	is the set of $m$ edges, determine a set $S \subseteq V$ of $k$ nodes that
	maximizes the number of edges covered by the nodes in $S$, i.e., incident on
	at least one of the nodes in $S$.
\end{prob} 

Being a generalization of the vertex cover problem, the $k$-VC problem is
NP-hard \cite{Bourgeois}. 
}
{}
\begin{thm}
	The \tumpg problem is NP-hard.
\end{thm}
\longv
{

\begin{proof}
	Given an instance of the $k$-VC problem, $G(V,E)$, we reduce it to an
	instance of \tumpg as follows. For each node $v_i \in V$, we create a
	base-station $B_i \in \mathcal{B}$. For each edge $(v_i,v_j) \in E$, we
	create a trajectory of size $2$, $\{ \langle B_i,\Delta_0,\eta_0 \rangle,
	\langle B_j,\Delta_0,\eta_0 \rangle \}$, where $\Delta_0$ and $\eta_0$ are
	arbitrary constant values for time interval and throughput respectively. Also,
	$\forall T_j, \forall B_i \in T_j, b_{ji}=1$, i.e., each of the
	base-stations incident on a trajectory acts as a bottleneck.

	If a subset $S \subseteq V$ is a solution to the $k$-VC problem, then the
	set of base-stations $B_S = \{B_i|v_i \in S\}$ is a solution to the \tumpg
	problem with $\gamma=0.5$. If the edge $(v_i,v_j)\in E$ is covered by the set $S$, then at
	least one of the two base-stations, $B_i, B_j \in B_S$ and, thus, the trajectory
	$\{ \langle B_i,\Delta_0,\eta_0 \rangle,$ $\langle B_j,\Delta_0,\eta_0
	\rangle \}$ becomes $0.5$-bottleneck-free, i.e., its utility becomes $1$.  Since
	the subset $S$ maximizes the number of edges covered by it, the solution
	$B_S$ maximizes the number of $0.5$-bottleneck-free trajectories or, in other
	words, maximizes the sum of trajectory utilities.

	Similarly, we  argue that if $B$ is a solution to the \tumpg problem with $\gamma=0.5$, then $S_B=\{v_i|B_i
	\in B\}$ is a solution to the $k$-VC problem. 	
	Since $\gamma=0.5$, the
	utility of a trajectory is maximized if and only if the trajectory becomes
	$0.5$-bottleneck-free, i.e., at least one of the two base-stations,
	$B_i,B_j$ on the trajectory, is in the solution $B$. This in turn implies
	that the edge $(v_i,v_j) \in E$ is covered by at least one of the nodes
	$v_i,v_j \in S$. Since $B$ maximizes the number of
	$0.5$-bottleneck-free trajectories, the subset $S_B$ maximizes the
	number of edges covered by it.
	
	Since the reduction requires space and time which is polynomial in the size of the
	input, the proof follows.
	\hfill{}
\end{proof}
}
{
  \textsc{Proof: } The proof entails a polynomial time reduction from the $k$-Vertex Cover problem \cite{Bourgeois}. The details are given in \cite{FullTUMP}.
}

\comment{
Since \tumpo and \tumpe are NP-hard, the generalized problem \tumpg is also
NP-hard.
}

\begin{figure}[t]
\begin{center}
\includegraphics[width=0.85\linewidth]{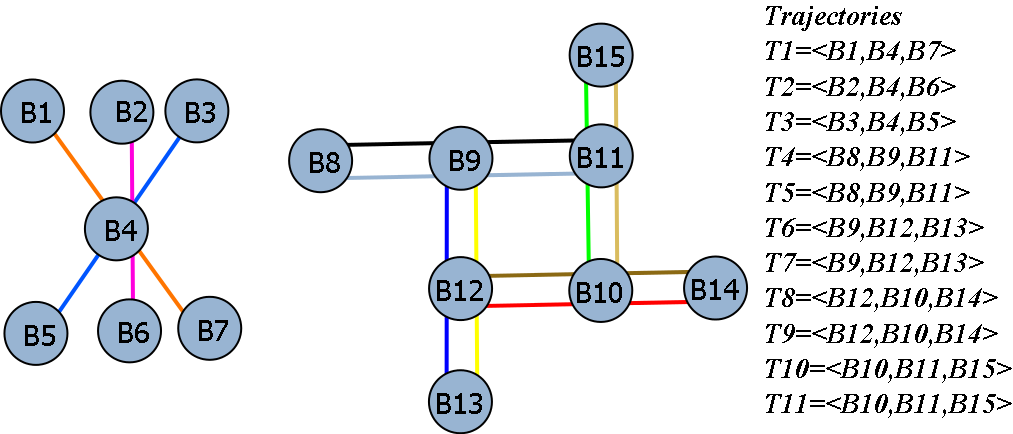}
\end{center}
\figcaption{Trajectories in Example~\ref{ex:one}.}
\label{fig:sampleEG}
\end{figure}



\section{Algorithms for $TUMP(\gamma)$}
\label{sec:algo}

\longv
{
This section describes the algorithms for the  \tumpg problem.
Firstly, we propose an integer programming based optimal algorithm, namely $IP_{\tumpg}$  and an approximation scheme based on LP relaxation of $IP_{\tumpg}$. Owing to their high running times, both of these algorithms are however, impractical for any reasonably sized dataset. The problem being NP-hard, we next present few approximation algorithms based on greedy paradigm, that not only offer faster running times, but also perform well in practice.

We begin by illustrating a simple example of \tumpg problem, that will be shortly evaluated by different algorithms.
}
{  
 This section describes the algorithms for the  \tumpg problem using the following illustrative example.
}

\comment{
Since \tumpg problem is NP-hard, in this section, we propose a few approximation algorithms to solve it. Before  proceeding further, we
state a working example, that will be evaluated by the proposed algorithms.
}
\comment{
Firstly, we propose an integer programming based optimal algorithm, namely $IP_{\tumpg}$  and an approximation scheme based on LP relaxation of $IP_{\tumpg}$. Owing to their high running times, both of these algorithms are however, impractical for any reasonably sized dataset. The problem being NP-hard, we next present three approximation algorithms based on greedy paradigm, that not only offer faster running times, but also perform well in practice.
}

\begin{example} \label{ex:one}	
	Fig.~\ref{fig:sampleEG} shows $11$ trajectories, $T_1,\cdots,T_{11}$, each of length $3$, passing through a
	set of $15$ base-stations $B_1,\cdots,B_{15}$.  We assume that each base-station is a bottleneck w.r.t. each
	of the trajectories incident on it. In addition, for ease of analysis, we assume that for a given
	trajectory, the $\Delta,\eta$ values are same for all the base-stations
	incident on it.
	\longv{	
	 For this reason, we did not list their
	values in Fig.~\ref{fig:sampleEG}.  This assumption eliminates the influence
	of $w_{ji}$'s on the solution. 
	}
	Thus, for any trajectory $T_j$ and a base-station $B_i \in T_j$, $w_{ji}=1/3$.
	\longv
	{	
	We next observe that while each of the
	base-stations $B_9,B_{10},B_{11}$ and $B_{12}$ have a maximum of $4$ incident trajectories,
	the base-stations $B_4$  has $3$ incident trajectories, the
	base-stations $B_8$, $B_{13},B_{14}$ and $B_{15}$ have $2$ incident trajectories and the rest
	have only $1$ trajectory passing through them.  	
	}			
	We set $k=3$ base-stations
	to upgrade as the budget parameter.
\longv
{	
	We evaluate this example for three different
	values of $\gamma$, as shown in Table~\ref{tab:utilities}. Following the above assumptions, we find that when $\gamma=0.33$ (resp. $\gamma=0.5$ and $\gamma=1$), it implies that at least one (resp. two and three) of the three base-stations on the trajectory must be upgraded, in order to make it $\gamma$-bottleneck-free. 
}
{
	We evaluate this example for two different
	values of $\gamma$, as shown in Table~\ref{tab:utilities}.
	When $\gamma=0.33$ (respectively,
	$\gamma=1$), it implies that at least $1$ (respectively, $3$) of the
	base-stations on the trajectory must be upgraded in order to make it
	$\gamma$-bottleneck-free. 
}		
The \emph{optimal} solutions for each of these cases are shown in the table.
\qed
\end{example}
\comment{
\begin{table}[t]\scriptsize
	\begin{center}
	\caption{Example of Trajectories, shown in Fig.~\ref{fig:sampleEG}.}
		\begin{tabular}{|c|c||c|c||c|c|}
			\hline
Traj. & Base-Stations & Traj. & Base-Stations & Traj. & Base-Stations \\
\hline
 $T_1$	& $B_1,B_4,B_7$	& $T_2$ & $B_2,B_4,B_6$ & $T_3$ & $B_3,B_4,B_5$ 	\\
 \hline
$T_4$ & $B_8,B_9,B_{11}$ & $T_5$ & $B_8,B_9,B_{11}$  & $T_6$ & $B_9,B_{12},B_{13}$\\   
 \hline
$T_7$ & $B_9,B_{12},B_{13}$& $T_8$ & $B_{12},B_{10},B_{14}$  & $T_9$ & $B_{12},B_{10},B_{14}$ \\
\hline 
$T_{10}$ & $B_{10},B_{11},B_{15}$ & $T_{11}$ & $B_{10},B_{11},B_{15}$  & &\\
\hline			
		\end{tabular}
		
		\label{tab:example}
\end{center}
\end{table}
}

We pose the \tumpg problem in
a graph setting.  Each instance of the \tumpg problem is associated with a
\emph{hyper-graph} $H = (V, E)$ where $V = \{v_1, \dots, v_n\}$ is the set of $n$ nodes
corresponding to the set of base-stations, $\mathcal{B}$, and $E = \{e_1, \dots,
e_m\}$ is the set of $m$ hyper-edges corresponding to the set of trajectories,
$\mathcal{T}$.  A node $v_i$ represents a base-station $B_i$ and a hyper-edge
$e_j$ represents the set of base-stations that the trajectory $T_j$ passes
through, i.e., $e_j = \{v_i | B_i \in T_j\}$. The degree of a node is the number of hyper-edges incident on it.

Given a set of nodes $S \subseteq
V$, its weight $w(S)$ denotes the number of hyper-edges, $e_j$, incident on at
least one of the nodes in $S$ such that $T_j$ is $\gamma$-bottleneck-free with
respect to the nodes in $S$.
For $\gamma=1$, $w(S)$ denotes the number of hyper-edges \emph{induced} on the
sub-hyper-graph formed by $S$, i.e., all the nodes of the hyper-edge
are contained within $S$. 
 
\comment{
For $\gamma=\epsilon$, $w(S)$ denotes the number of
hyper-edges \emph{covered} by the set $S$, i.e., incident on at least one of the
nodes in $S$. 
}
 Referring to this hyper-graph model, henceforth, we shall use the
terms node and base-station (and respectively, hyper-edge and trajectory)
interchangeably.  The example in Fig.~\ref{fig:sampleEG}  shows this hyper-graph setting.

\longv
{
As a pre-processing step to all our algorithms, we discard the trajectories that
cannot be made $\gamma$-bottleneck-free by upgrading any set of $k$
base-stations.  If the sum of weights of the $k$ bottleneck base-stations with
the largest weights is less than $\gamma$, the trajectory can be deemed as
infeasible and can be pruned.
}
{
The \tumpg problem can be solved \emph{optimally} using an integer linear
program (ILP), detailed  in \cite{FullTUMP}. However, since the problem is
NP-hard, its exponential running time is impractical, given the large number of
trajectories and base-stations that a typical cellular operator has to consider.
Further, we consider a Linear Programming (LP) relaxation based
approach~\cite{lprelax}, which is a common technique to approximately solve
ILPs. However, as we show in Section~\ref{sec:optComp}, the running time is
still impractical. Thus, we design greedy approximation algorithms that
perform well in practice.
}

\longv{
\subsection{Optimal Algorithm}
The \emph{optimal} solution to the \tumpg problem is represented by the following integer programming formulation, denoted by $IP_{\tumpg}$:
\begin{align}
	\label{eq:ip}
	\max w_{\gamma} &= \sum_{j=1}^m U_j \nonumber \\
	\text{s.t. } \sum_{i=1}^n x_i &\leq k, \nonumber \\
	\forall i = 1, \dots, n, \ &\forall j = 1, \dots, m, \quad
	x_i, U_j \in \{0,1\}, \nonumber \\
	\text{and } &\forall j = 1, \dots, m, \quad
	U_j \leq \frac{W_j}{\gamma} \nonumber  
\end{align}

	where  $W_j$  is given by Eq.~\ref{eq:traj_weight}, assuming it to be the bottleneck utility function. The solution $\{\tilde{x}, \tilde{U}\}$ of the above IP formulation yields the
\emph{optimal} weight $\hat{w}_{\gamma}$ (i.e., the number of
$\gamma$-bottleneck-free trajectories) to the \tumpg problem. Since the utility of each trajectory
is at most $1$, necessarily, $\tilde{w}_{\gamma} \leq m$, where $m$ is the total
number of trajectories. 
\comment{
Considering Example~\ref{ex:one}, the optimal selection of base-stations and corresponding utilities for different values of $\gamma$, are listed in Table~\ref{tab:utilities}. 
 }
Since obtaining the optimal solution requires time that
is exponential in the input size, it is
impractical.
Approximation algorithms are, thus, required to solve the problem in practical
running times. 
}

\comment{
 If any algorithm $A$ solving \tumpg \emph{always} returns a
feasible solution with weight $w_{\gamma}^A \geq r . \tilde{w}_{\gamma}$ (for
some fixed $r \in [0,1]$), then $r$ is said to be the \emph{approximation bound}
for the algorithm $A$.
Moreover, if an algorithm has an approximation bound of $r$ for the $\tumpo$
problem, then $r$ acts as an approximation bound for any $\tump(\gamma < 1)$ as
well because \tumpo is the worst case instance for \tumpg class of problems.

}
\longv
{
\subsection{Approximation Algorithms}

In the following sections, we present few approximation algorithms for the \tumpg  problem.  
}
{}

\longv
{
\begin{table*}[t]\scriptsize
	\begin{center}
	\caption{Utilities derived from different algorithms for Example~\ref{ex:one} with $k = 3$.}
		\label{tab:utilities}
		\begin{tabular}{|c||c|c||c|c||c|c|}
			\hline
			Algorithms & \multicolumn{2}{|c|}{$\gamma = 0.33$}& \multicolumn{2}{|c|}{$\gamma = 0.5$}& \multicolumn{2}{|c|}{$\gamma = 1$} \\
\hline			
	 	& Upgrades & Utility & Upgrades & Utility & Upgrades & Utility \\	
			\hline
			\hline
			Optimal & $B_4,B_{11},B_{12}$ & $11$ & $B_{10},B_{11},B_{12}$ & $4$ & $B_{10},B_{12},B_{14}$ & $2$ \\
			\hline
			\simg   & $B_{10},B_{11},B_{12}$ & $8$ & $B_{10},B_{11},B_{12}$ & $4$ & $B_{10},B_{11},B_{12}$ & $0$  \\
			\hline
			\incg   &  $B_{4},B_{11},B_{12}$& $11$ &  $B_{10},B_{11},B_{12}$ & $4$ & $B_{10},B_{11},B_{12}$ & $0$ \\
			\hline
			\decg   & $B_{4},B_{11},B_{12}$ & $11$ & $B_{10},B_{11},B_{12}$& $4$ & $B_{10},B_{12},B_{14}$& $2$ \\
			\hline
		\end{tabular}	
	\end{center}
\end{table*}
}
{
 \begin{table}[t]\scriptsize
	\begin{center}
	\caption{Utilities derived from different algorithms for Example~\ref{ex:one} with $k = 3$.}
		\label{tab:utilities}
		\begin{tabular}{|c||c|c||c|c|}
			\hline
			Algorithms & \multicolumn{2}{|c|}{$\gamma = 0.33$}&  \multicolumn{2}{|c|}{$\gamma = 1$} \\
\hline			
	 	& Upgrades & Utility & Upgrades & Utility \\	
			\hline
			\hline
			Optimal & $B_4,B_{11},B_{12}$ & $11$ &  $B_{10},B_{12},B_{14}$ & $2$ \\
			\hline
			\simg   & $B_{10},B_{11},B_{12}$ & $8$ &  $B_{10},B_{11},B_{12}$ & $0$  \\
			\hline
			\incg   &  $B_{4},B_{11},B_{12}$& $11$ &   $B_{10},B_{11},B_{12}$ & $0$ \\
			\hline
			\decg   & $B_{4},B_{11},B_{12}$ & $11$ &  $B_{10},B_{12},B_{14}$& $2$ \\
			\hline
		\end{tabular}	
	\end{center}
\end{table}
}

\longv
{
Let $A$ be any approximation algorithm  for \tumpg that \emph{always} returns a
feasible solution $S$ with weight $w(S) \geq r . w(OPT)$, for
some fixed $r \in [0,1]$, where $OPT$ refers to an \emph{optimal} solution . Then $r$ is said to be the \emph{approximation bound}
of the algorithm $A$.
Moreover, if an algorithm has an approximation bound of $r$ for the $\tumpo$
problem, then $r$ acts as an approximation bound for any $\tump(\gamma < 1)$ as
well, because \tumpo is the worst case instance of \tumpg. 
}
{}

\comment{
\subsubsection{\randk}

The algorithm \emph{randomly} selects a set $S \subset V$ of size $k$.  As we
analyze later, the approximation bounds for this algorithm are not good enough
in practice; hence, we propose algorithms based on the greedy paradigm and
linear programming.
}

\longv
{
\subsection{\lpr}

Here, we propose a \emph{linear programming (LP)} based heuristic for the
\tumpg problem.  The \lpr solution is based on the LP relaxation of the integer programming
formulation $IP_{\tumpg}$ specified for the optimal solution.  This relaxation allows the
variables $x_i, U_j$ to take fractional values, i.e., $\forall i, \ \forall j, \
0 \leq x_i, U_j \leq 1$.  Assume that $S^* = \{x_i^*,U_j^*\}$ denotes the
(optimal) solution to the above linear program.  We derive an integer
\emph{approximate} solution $\hat{S} = \{\hat{x_i}, \hat{U_j}\}$ from $S^*$ as follows.
We pick the highest $k$ values from $(x_1^*, \dots, x_n^*)$ and
set the corresponding $\hat{x_i}$ to $1$.  
\comment{
Although this technique is simple, it is hard to analyze its approximation bound because
the probability that $\hat{x_i} = 1$ depends not only on the value of $x_i^*$
but its relative position in the ordering of the values $(x_1^*, \dots, x_n^*)$. 
}

The expensive step of this approach is solving the linear program that involves $O(m+n)$ constraints and $O(m+n)$ variables.
The running time of this approach, is although  considerably less than the optimal algorithm based on integer programming ($IP_{\tumpg}$), but still impractical for any reasonably sized dataset, as shown in the experiments, discussed in Section~\ref{sec:results}. Therefore, we next, propose a set of approximation algorithms based on greedy paradigm, that not only offer faster running times, but also perform well in practice. 
}
{
}

\comment{
Thus, we apply the technique of \emph{randomized rounding}, where the sampling
process is independent of the order of $x_i^*$'s.  Every $\hat{x_i}$ is set to
$1$ with probability $p_i$:
\begin{align}
	\label{eq:prob_value}
	\forall i, \ p_i = \frac{k}{\sum_{\forall i} (x_i^*)^\phi}. (x_i^*)^\phi
	\text{ where } \phi = (\dmax)^{-\frac{\gamma-\epsilon}{1-\epsilon}}
\end{align}

Since $0<\epsilon \le \gamma \le 1$, we get $0 \le \phi \le 1$.  Further, since
$\sum_i x_i^* = k$ and $0 \leq x_i^* \le 1$, we also get $\sum_i (x_i^*)^{\phi}
\ge k$.  Thus, $0 \le p_i \le 1$.

This essentially implies that the node $v_i \in V$ is sampled with the
probability $p_i$. We have chosen the value of $p_i$ in such a way that if
$x_i^* \le x_j^*$, then $p_i \le p_j$.  The rationale behind choosing the
precise value of $p_i$ will be clear during the analysis in
Section~\ref{sec:bounds}.  
 
Since the sampling process is independent, it is not guaranteed that $|\hat{S}|
= \sum_{i=1}^n \hat{x_i} = k$.  However, it is important to note that
irrespective of the value of $\gamma$, the expected size of the sample, i.e.,
$E[|\hat{S}|] = \sum_i p_i = k$.

If there are less than $k$ nodes in $\hat{S}$, i.e., $|\hat{S}| < k$, we use
\incg to add $k-|\hat{S}|$ nodes from the set $V \setminus \hat{S}$.  If
there are more nodes, i.e., $|\hat{S}| > k$, we employ \decg to reduce the size
to $k$.  The final solution obtained is denoted by $\hat{S^*}$.
}

\longv
{}
{ 
}

\subsection{\simg}

For  each base-station $B_i\in \mathcal{B}$, we define its bottleneck-weight, $\omega_i=\{\sum_{j=1}^m w_{ji}|b_{ji}=1\}$. Typically, a base-station that acts as a bottleneck for a large number of trajectories for a considerable fraction of their total time will have a high bottleneck-weight, and is thus, a good candidate for upgradation.
 The
simple greedy approach picks the $k$ base-stations having the \emph{largest
bottleneck-weights}.

\begin{example}
Consider Example~\ref{ex:one}. Since each of the base-stations $B_9,\dots,B_{12}$ have maximal bottleneck-weight $4/3$, this approach will select  the base-stations $B_{10},B_{11}$ and $B_{12}$ (breaking ties on higher index of base-station), irrespective of the value of $\gamma$. The utilities, thus obtained for different values of $\gamma$, are listed in Table~\ref{tab:utilities}.
\qed
\end{example}

The primary drawback of
this approach is that it is independent of the notion of  trajectory utilities. This is why   
 it does not perform well, especially when $\gamma$ is high. Nevertheless, owing
 to its simplicity, we consider it as the baseline algorithm for \tumpg, and
 compare the performance of other algorithms against it, as
 discussed in detail in Section~\ref{sec:results}.\\

\longv
{
\subsubsection*{Analysis of \simg}
First, we analyze the time and space complexity of \simg, and then analyze its approximation bound.

\begin{thm}
The time and space complexity bounds for \simg are $O(m\dmax+n\log k)$ and $O(m\dmax)$ respectively, where $m$ is the total number of trajectories, \dmax is the maximum length of any trajectory, $n$ is the total number of base-stations, and $k$ is the budget parameter.
\end{thm}
\begin{proof}
First we analyze the time complexity. Given that there are $m$ trajectories, and each of them has length at most $\dmax$, scanning the input and computing the bottleneck-weight of each base-station, takes $O(m\dmax)$ time. Further, computation of top-$k$ base-stations w.r.t. their bottleneck-weights takes $O(n\log k)$ time \cite{ArtOfProgrammingVol3} where $n$ is the total number of base-stations and $k$ is the number of base-stations to be upgraded. Thus, the time complexity of \simg is $O(m\dmax+n\log k)$. 

As we analyze the space complexity, assuming each trajectory is stored as a list of tuples $\Phi$ taking $O(\dmax)$ space, the total input takes $O(m\dmax)$ space. We require $O(n)$ space for computation of the bottleneck-weights ( $O(1)$ space for each node). Since $md \ge n$, the  space complexity, is thus $O(m\dmax)$.
\end{proof}
}
{
}  
\longv
{

Unfortunately, this algorithm has \emph{no  approximation
bound} for the \tumpg problem, as shown in Table~\ref{tab:summary of analysis}. 
\begin{thm}
	\label{thm:sim_tumpo}
	\simg has no bounded approximation for \tumpg.
\end{thm}
\begin{proof}
	To show that \simg does not offer any approximation bound, we	consider the following instance of \tumpg problem. Let $\mathcal{B}=\{B_1,\dots,B_n\}$ be the set of  base-stations, and $\mathcal{T}=\{T_1,\dots,T_m\}$ be the set of trajectories such that $m= \lfloor \frac{n}{2} \rfloor +1$. For $j=1,\dots,m-1$, the trajectory $T_j$ of length $2$, passes through the base-stations $B_j$ and $B_{j^\prime}$ where $j^\prime = j+\lfloor \frac{n}{2}\rfloor$.  The trajectory $T_m$ of length $\lfloor \frac{n}{2}\rfloor$, passes through the base-stations $B_1,\dots,B_{\lfloor \frac{n}{2}\rfloor}$. We assume that each base-station is a bottleneck for each of the trajectories passing through it. Additionally, for any trajectory, we assume that the $\triangle,\eta$ values are same across all the base-stations that the trajectory passes through. Therefore, we note that for any trajectory $T_j$, $j=1,\dots,m-1$, that passes through the base-stations $B_j$ and $B_{j^\prime}$, $w_{jj}=w_{jj^\prime}=0.5$. For trajectory $T_m$, for $i=1,\dots,\lfloor \frac{n}{2}\rfloor$, we note that $w_{mi}=1/\lfloor \frac{n}{2}\rfloor$. Therefore the bottleneck weight $\omega_i$ for base-station $B_i$ is given by:
\begin{align}	
	\omega_i = 
			\begin{cases}
				0.5 + 1/\lfloor \frac{n}{2}\rfloor & \text{ for } i=1,\dots,\lfloor\frac{n}{2}\rfloor  \\
				0.5 & \text{ for } i=\lfloor\frac{n}{2}\rfloor+1,\dots,n 
			\end{cases}	
\end{align}
Thus, for any $2 \le k \le \lfloor\frac{n}{2}\rfloor$, \simg would select any $k$ base-stations from the set $B_1,\dots,B_{\lfloor\frac{n}{2}\rfloor}$, due to their higher bottleneck weight. But any such selection cannot make any trajectory $\gamma$-bottleneck-free for $\gamma=1$. On the other hand, an optimal solution would make at least one trajectory $\gamma$-bottleneck-free. Thus, we show that \simg has no approximation bound for \tumpg.
	\hfill{}
\end{proof}
The time and space complexity bounds of the proposed algorithms are stated in  Table~\ref{tab:summary of analysis}. The proofs are provided in \cite{FullTUMP}. Further, we show that this algorithm has \emph{no approximation bound} for the \tumpg problem, \cite{FullTUMP}.
}
{
}
\subsection{\incg}
Based on the principle of \emph{maximizing marginal gain}, this approach starts
with an empty set of nodes $S_0=\varnothing$, and incrementally adds nodes
such that each successive addition of a node produces the maximal marginal
gain in the weight of the solution.  The algorithm proceeds in iterations
$\theta=\{1,\dots,k\}$.  In the beginning of iteration $\theta$, suppose the
existing solution is the set of nodes $S_{\theta-1}$ with weight
$w(S_{\theta-1})$.  The node $v_{\theta}$ from the remaining set $V \setminus
S_{\theta-1}$ is added such that $w(S_{\theta-1} \cup \{v_\theta\})$ is maximal.
The new set is referred to as $S_\theta$.

In any iteration, if multiple candidate nodes have the same 
maximal marginal utility, we select the one with the largest bottleneck-weight.
Still, if ties remain, then without loss of generality, we
break the tie by selecting the node with the highest index (the indices are
arbitrary but unique).

\begin{example}
We first evaluate \incg on Example~\ref{ex:one} for $\gamma=0.33$. At iteration $1$, nodes $B_9,\dots,B_{12}$  have the same maximal marginal utility of $4$ and the
same bottleneck-weight $4/3$. So, we pick $B_{12}$ as it has the highest index.
Next, we select $B_{11}$, as it offers the maximal marginal utility of $4$.  The
base-station $B_{10}$ does not offer the maximal marginal gain any more.  The
base-station $B_4$ with marginal gain of $3$ becomes the best choice next. Thus, we obtain a net utility of $11$, which equals the \emph{optimal} utility. 

\longv
{
The selections and utilities for other values of $\gamma$ case, are shown in Table~\ref{tab:utilities}. Noticeably, for $\gamma=1$, this algorithms offers $0$ utility.
}
{
However, as shown in Table~\ref{tab:utilities}, this algorithm offers $0$ utility for the case $\gamma=1$. 
}
  Evaluating this case, in iteration $1$, we find that all the base-stations offer $0$ marginal utility, and so we sample the base-stations $B_9,\dots,B_{12}$ on the basis of maximal bottleneck-weight. Eventually, we pick $B_{12}$ owing to its highest index.
  In the following $2$ iterations, once again, we find that the marginal utility
  offered by any  base-station is $0$. Respecting the tie-breaking criteria, we
  pick $B_{11}$ and $B_{10}$ in respective order. This strategy, however, does
  not make any trajectory $\gamma$-bottleneck-free, and yields $0$ utility.
  \qed
\end{example}

 The
  next algorithm (\decg) addresses this limitation.

\longv
{
\subsubsection*{Analysis of \incg}
Fist we analyze the complexity bounds, followed by the approximation bound.

\begin{thm}
The time and space complexity bounds for \incg are $O(m\dmax^2)$ and $O(m\dmax)$ respectively.
\end{thm}

\begin{proof}
We first analyze the time complexity of the algorithm. As a pre-processing step, the algorithm computes the bottleneck-weight and initial marginal utility for each node, and computes the initial trajectory utility $W_j$ for each trajectory. This step takes $O(m\dmax)$ time. 
\comment{
With two scans of the input, the algorithm computes the bottleneck-weight of each node. This step takes $O(m\dmax)$ time. In each iteration, the algorithm keeps track of the current weight $W_j$ of trajectory $T_j$. Additionally, 
for each iteration $\theta$, for each node $v_i \in V\setminus S_{\theta-1}$, the algorithm keeps track of its (current) marginal utility w.r.t. the current solution $S_{\theta-1}$. 
Before iteration $1$, we compute the marginal utilities for each node $v_i \in V$ which takes $O(m\dmax)$ time. 
}
Next, in  any iteration $\theta$, we select the node $v_\theta$ that has maximal marginal utility (applying the tie-breaking rules, if necessary), and add it to the set $S_\theta$. This step takes $O(n)$ time. Then  for each trajectory $T_j$ incident on $v_\theta$, for each $v_i \in V \setminus S_\theta$ such that $B_i \in T_j$, and $b_{ji}=1$, we check if adding $v_i$ to $S_\theta$ would make $T_j$ $\gamma$-bottleneck-free. If so, we increment the current marginal utility of $v_i$ by $1$.
If $\delta_\theta$ is the degree of the node $v_\theta$, then this step takes $O(\delta_\theta \dmax)$ time. Since $\sum_{i=1}^n \delta_i \le m\dmax$, hence the total time spent in the above step over $k$ iterations, is $O(m\dmax^2)$. 
\comment{
 Taking advantage of the current marginal utilities, 
this approach avoids re-computation of marginal utilities for each node in every iteration.  
}
The total time complexity of the algorithm, running over $k$ iterations, is thus $O(m\dmax^2 +kn)$. 

To analyze the space complexity, we find that besides the input that takes $O(m\dmax)$ space, the algorithm requires $O(1)$ space for each of the trajectories to store their current utilities that gets updated after each iteration. Further, we need $O(1)$ space for each of the nodes to store their marginal utility,  bottleneck-weight, and to hold the information if they are selected. Thus, the space complexity of the algorithm is $O(m\dmax+m+n)=O(m\dmax)$.
\end{proof}

The time complexity of this algorithm can be further improved to $O(m\dmax^2 + k\log n)$ by using advanced data structures such as Fibonacci Heaps \cite{FibonacciHeap} to store the marginal utilities of the nodes. However since the total time for updating of marginal utilities of the nodes  $(O(m\dmax^2))$ dominates the time for computing the node with maximal marginal gain ($O(k\log n)$), we avoided this implementation.
}
{
}

\longv
{
Unfortunately, similar to \simg, \incg has \emph{no 
approximation bound} for \tumpg.
\begin{thm}
	\label{thm:inc_tumpo}
	\incg has no bounded approximation for \tumpg.
\end{thm}
\begin{proof}
	To show that \incg does not offer any approximation bound, we	consider the following instance of \tumpg problem. Let $\mathcal{B}=\{B_1,\dots,B_n\}$ be the set of  base-stations, and $\mathcal{T}=\{T_1,T_2\}$ be the set of trajectories. Without loss of generality, let $n$ be a even number. Let trajectory $T_1$ pass through the base-stations with odd index, i.e., $B_1,B_3,B_5,\dots,B_{n-1}$, and let trajectory $T_2$ pass through the base-stations with even index, i.e., $B_2,B_4,B_6,\dots,B_{n}$.	
We assume that each base-station is a bottleneck for each of the trajectories passing through it. Additionally, for any trajectory, we assume that the $\triangle,\eta$ values are same across all the base-stations that the trajectory passes through. Since there is only a single trajectory (of length $n/2$) incident on each base-station, the bottleneck weight $\omega_i$, for any base-station $B_i \in \mathcal{B}$, is $2/n$. Let $k=n/2$ and $\gamma=1$. 

Now let us evaluate \incg on this instance. In each iteration, $1,\dots,k$,  the marginal utility of each node is $0$. Respecting the tie-breaking criteria, in iteration $1,\dots,k$, we select the nodes $B_n,B_{n-1},\dots,B_{\lfloor\frac{n}{2}\rfloor +1}$, respectively. However, by this selection, neither of the trajectories could become $\gamma$-bottleneck-free. 
 On the other hand, an optimal solution would select all the $n/2$ base-stations on either of trajectories, $T_1$ or $T_2$, thus making at least one trajectory $\gamma$-bottleneck-free. Thus, we show that \incg has no approximation bound for \tumpg.		%
	\hfill{}
\end{proof}
The time and space complexity bounds of this algorithm are stated in  Table~\ref{tab:summary of analysis}. The  proofs are provided in \cite{FullTUMP}. Further,  we show that this algorithm has \emph{no 
approximation bound} for the \tumpg problem, \cite{FullTUMP}.
}
{
}

\subsection{\decg}

This algorithm operates in the reverse order by \emph{minimizing marginal loss}.  It starts with the full set of nodes $V$ and
successively removes nodes in a manner that minimizes the marginal loss in
the weight of the resulting set.  More precisely, it starts with $S_0=V$, and
removes one node in each iteration $\theta=\{1,\dots,n-k\}$.  At the start of
iteration $\theta$, suppose the existing set of nodes is $S_{\theta-1}$ with
weight $w(S_{\theta-1})$.  From this set, the node $v_\theta$ is removed such
that $w(S_{\theta -1} \setminus \{v_\theta\})$ is maximal.  The new set is
referred to as $S_\theta$.

Moreover, after each iteration, all trajectories that can no longer be made
$\gamma$-bottleneck-free are pruned.  In any given iteration, if multiple
candidate nodes qualify to be deleted, then the one with the smallest
bottleneck-weight is chosen. Still, if there are multiple candidates, the tie is
broken by removing the node with the lowest index.

\begin{example}  
We first evaluate \decg on Example~\ref{ex:one} for $\gamma=1$.  At iteration
$1$, any of the base-stations, $B_1,B_2,B_3,B_5,B_6$ and
$B_7$, have the same minimal marginal loss in utility of $1$. Respecting the tie-breaking
criteria, we prune $B_1$. Consequently, trajectory $T_1$ becomes infeasible and
is, thus, pruned. This results in marginal utility of $B_7$ to become $0$, as
there is no other incident trajectory on $B_7$. So, in the next iteration, we prune
$B_7$. Proceeding in this manner, the order of deletions of the nodes in
subsequent iterations is as shown in Table~\ref{tab:node-deletions by decg}. At
the end of 12 iterations, we are left with the base-stations $B_{10},B_{12}$ and
$B_{14}$ which render 2 trajectories, $T_8$ and $T_9$, $\gamma$-bottleneck-free,
which is same as the optimal algorithm.

\longv
{
 Table~\ref{tab:utilities} shows the output of \decg for other values of $\gamma$, along with Table~\ref{tab:node-deletions by decg} which shows the order of deletions of nodes in the $n-k=12$ iterations. 
}
{
 Table~\ref{tab:utilities} shows the output of \decg for the case $\gamma=0.33$, along with Table~\ref{tab:node-deletions by decg} which shows the order of deletions of nodes in the $n-k=12$ iterations.
}
\qed
\end{example}

\longv{
\begin{table}[t]\scriptsize
	\begin{center}
	\caption{Order of node-deletions by \decg for  Example~\ref{ex:one}.}
		\label{tab:node-deletions by decg}
		\begin{tabular}{|c||c|}
			\hline
			$\gamma$ & Order of Nodes deleted during iteration 1 to 12 $\rightarrow$ \\
\hline		
$0.33$ &	 $B_1,B_2,B_3,B_5,B_6,B_7,B_8,B_{13},B_{14},B_{15},B_9,B_{10}$\\	
			\hline
	$0.5$ & $B_1,B_2,B_3, B_8, B_{13},B_{14},B_{15},B_5,B_6,B_4,B_7,B_9$ \\
			\hline			
			$1$ & $B_1,B_7,B_2,B_6,B_3,B_4,B_5,B_8,B_9,B_{13},B_{11},B_{15}$\\
			\hline						
		\end{tabular}
		
	\end{center}
\end{table}
}
{
  \begin{table}[t]\scriptsize
	\begin{center}
	\caption{Order of node-deletions by \decg for Example~\ref{ex:one}.}
		\label{tab:node-deletions by decg}
		\begin{tabular}{|c||c|}
			\hline
			$\gamma$ & Order of Nodes deleted during iteration 1 to 12 $\rightarrow$ \\
\hline		
$0.33$ &	 $B_1,B_2,B_3,B_5,B_6,B_7,B_8,B_{13},B_{14},B_{15},B_9,B_{10}$\\	
			\hline
			$1$ & $B_1,B_7,B_2,B_6,B_3,B_4,B_5,B_8,B_9,B_{13},B_{11},B_{15}$\\
			\hline						
		\end{tabular}
		
	\end{center}
\end{table}
}


\longv
{
\subsubsection*{Analysis of \decg}
The time complexity analysis of \decg is similar to that of  \incg, with the difference that in this case there are $n-k$ iterations. Thus the time complexity  is $O(m\dmax^2+(n-k)n)$, as shown in Table~\ref{tab:summary of analysis}. Since, $k$ is likely to be less that $n/2$, \decg is expected to take more number of iterations, and is therefore, more expensive in terms of time, when compared to \incg.
The space complexity of \decg is same as in the case of \incg, i.e., $O(m\dmax)$.

We next analyze the approximation bound of \decg algorithm.
\begin{thm}
	\label{thm:dec_tumpo}
	The approximation bound of \decg for $\tumpg$ is ${\binom{k}{\dmax}} /
	{\binom{n}{\dmax}}$.
\end{thm}

\begin{proof}
	To analyze the worst case scenario, we assume that each base-station is a bottleneck w.r.t. each of the incident trajectories, and further, $\gamma=1$.
	Assume that $S_\theta$ denotes the set of nodes in the sub-hyper-graph
	resulting at the end of iteration $\theta =\{1,\dots,n-k\}$.  Since
	$\gamma=1$, $w(S_\theta)$ denotes the number of hyper-edges \emph{induced} in this
	sub-hyper-graph, i.e., all its nodes must be in $S_\theta$.  Following the above assumption,  once the node $v_\theta$ is removed, all its
	incident hyper-edges in $S_{\theta-1}$ are removed, because the
	corresponding trajectories can no longer become $\gamma$-bottleneck-free.  Therefore,
	the node $v_\theta$ (selected to be pruned due to its minimal marginal loss) must have the minimal degree in the
	sub-hyper-graph induced over the nodes in $S_{\theta-1}$. This ensures
	that the weight of the resulting set $w(S_\theta) = w(S_{\theta-1} \setminus
	\{v_\theta\})$ is maximal.

Note that the sum of degrees of
	nodes in $S_{\theta-1}$ is equal to the sum of the lengths of the
	trajectories induced in the sub-hyper-graph $S_{\theta-1}$ which is at most
	$w(S_{\theta-1})\dmax$ where $\dmax$ is the maximum length of any
	trajectory.  Since $|S_{\theta-1}| = n - \theta + 1$, the average degree of
	a node in $S_{\theta-1}$ is at most $\frac{w(S_{\theta-1})\dmax}{|
	S_{\theta-1}|}=\frac{w(S_{\theta-1})\dmax}{n-\theta+1}$.  After \decg
	removes the node $v_\theta$ with  minimal degree (which is at most the average
	degree), the weight of the sub-hyper-graph $S_\theta$ is bounded as
	%
\begin{small}
	\begin{align*}
		w\left(S_\theta\right) &\ge w\left(S_{\theta -1}\right)
			- w\left(S_{\theta -1}\right) \frac{\dmax}{n-\theta +1}
	\end{align*}
	\end{small}
	Hence, after $n-k$ iterations
	%
\begin{small}
	\begin{align*}
	w\left(S_{n-k}\right) & \ge w(S_0) \prod_{\theta=1}^{n-k}
		\left(1-\frac{\dmax}{n-\theta +1}\right)
	\end{align*}
	\end{small}

	For $\tumpo$, we can assume that $k \geq \dmax$ since any
	trajectory with $|T_j|>k$ can be pruned as it can never be made
	$\gamma$-bottleneck-free.  Hence, using $n \geq k \geq \dmax$ and $w(S_0) = m$,
	%
\begin{small}
	\begin{align*}
	w\left(S_{n-k}\right) &\geq m \frac{k(k-1)\dots(k-\dmax+1)}{n(n-1)\dots(n-\dmax+1)}
		= m\left(\frac{\binom{k}{\dmax}}{\binom{n}{\dmax}}\right)
	\end{align*}
\end{small}
	Since weight of any optimal solution is at most $m$, the approximation bound of \decg is ${\binom{k}{\dmax}} / {\binom{n}{\dmax}}$.
	\hfill{}
\end{proof}
}
{
}
%
\longv
{
}
{
\subsection{Properties of the Approximation Algorithms}
The  time and  space complexities and the approximation bounds of the proposed
algorithms are summarized in Table~\ref{tab:summary of analysis}. The proofs are
available in \cite{FullTUMP}. Noticeably, \simg and \incg have \emph{no
approximation bounds}. The approximation bound for \decg is, however, non-zero.
The following theorem provides the analysis.

\begin{thm}
	\label{thm:dec_tumpo}
	The approximation bound of \decg for $\tumpg$ is ${\binom{k}{\dmax}} /
	{\binom{n}{\dmax}}$.
\end{thm}
\begin{proof}
	To analyze the worst case scenario, we assume that each base-station is a bottleneck w.r.t. each of the incident trajectories, and further, $\gamma=1$.
	Since
	$\gamma=1$, $w(S_\theta)$ denotes the number of hyper-edges \emph{induced} in this
	sub-hyper-graph, i.e., all its nodes must be in $S_\theta$.  Once the node $v_\theta$ is removed, all its
	incident hyper-edges in $S_{\theta-1}$ are removed, because the
	corresponding trajectories can no longer become $\gamma$-bottleneck-free.  Therefore,
	the node $v_\theta$ (selected to be pruned) must have the minimal degree in the
	sub-hyper-graph induced over the nodes in $S_{\theta-1}$. This ensures
	that the weight of the resulting set $w(S_\theta) = w(S_{\theta-1} \setminus
	\{v_\theta\})$ is maximal.

Note that the sum of degrees of
	nodes in $S_{\theta-1}$ is equal to the sum of the lengths of the
	trajectories induced in the sub-hyper-graph $S_{\theta-1}$ which is at most
	$w(S_{\theta-1})\dmax$.  Since $|S_{\theta-1}| = n - \theta + 1$, the average degree of
	a node in $S_{\theta-1}$ is at most $\frac{w(S_{\theta-1})\dmax}{n-\theta+1}$.  After \decg
	removes the node $v_\theta$ with  minimal degree (which is at most the average
	degree), the weight of the sub-hyper-graph $S_\theta$ is bounded as
$	
		w(S_\theta) \ge w(S_{\theta -1}) (1
			- \frac{\dmax}{n-\theta +1}).
$	
	Hence, after $n-k$ iterations,
$
	w(S_{n-k})  \ge w(S_0) \prod_{\theta=1}^{n-k}
		(1-\frac{\dmax}{n-\theta +1}).
$	

	As $\gamma=1$, we  assume that $k \geq \dmax$ since any
	trajectory with $|T_j|>k$ can be pruned as it can never be made
	$\gamma$-bottleneck-free.   Using $k \ge \dmax$, and $w(S_0) = m$, we get
$
	w(S_{n-k}) \geq  m(\binom{k}{\dmax}/\binom{n}{\dmax}).
$	
	As the weight of any optimal solution is at most $m$, the proof follows. (The detailed proof is given in \cite{FullTUMP}.)
	\hfill{}	
\end{proof}
}

\subsection{Equivalence of Incremental and One Shot Upgrades}
An interesting and very useful property of the proposed \incg and \decg
algorithms is that they naturally support incremental upgrades of base-stations.
Note that \incg (respectively, \decg) selects (respectively, prunes) one base-station in each iteration, and the selection criteria is independent of the budget parameter $k$.
Hence, it can be shown that, for both these algorithms, piecewise incremental upgrades
is equivalent to a one-shot upgrade, provided the total budget is same in both
the cases.

More formally, if $k_1 + k_2 + \dots = k\ (\forall k_i > 0)$, then successive
applications of \incg (respectively \decg) with budget parameters $k_1$,
followed by $k_2$, etc., would upgrade the same set of base-stations as a
one-shot application of \incg
(respectively, \decg) would with budget parameter $k$.
\longv
{}
{The formal proofs are available in \cite{FullTUMP}.
}
This property is very important as it allows
the network planners to upgrade as and when some
budget is allocated. The overall effect on the network is the same even if the
entire budget was made available at one go. 
\longv
{
First we establish this claim for \incg, followed by \decg.
\begin{prop}
	\label{claim:inc}
	Given any instance of \tumpg,
	\begin{align}
		S_k = S_\theta \cup S'_{k-\theta},\ \forall 1\leq k\leq n,\, 1\leq \theta \leq k
	\end{align}
	where $S_k$ is the set of nodes (corresponding to the upgraded
	base-stations) by \incg($k$), and $S_\theta$, $S'_{k-\theta}$ are the sets
	of nodes provided by two consecutive operations of \incg($\theta$) on
	$V$, and \incg($k-\theta$) on $V \setminus S_\theta$ respectively.
\end{prop}

\begin{proof}
	Observe that the algorithm always selects the nodes in $V$ in a
	consistent order. This follows from the fact that the algorithm breaks all
	ties deterministically (if needed, using the index of the nodes).
	Therefore, without loss of generality, assume that \incg($n$) applied on the
	node set $V$ selects the nodes in the order $v_1,\dots,v_n$. Hence,
	\incg($k$) applied on the node set $V$ would select the nodes $S_k =
	\{v_1,\dots,v_{k}\}$.  For $\theta \leq k$,
	$S_\theta=\{v_1,\dots,v_{\theta}\} \subseteq S_{k}$. Applying
	\incg($k-\theta$) on the set $V \setminus S_{\theta} =
	\{v_{\theta+1},\dots,v_{n}\}$ would select the nodes
	$S'_{k-\theta}=\{v_{\theta+1},\dots,v_{k}\}$. Hence, $S_{k} = S_{\theta}\cup
	S'_{k-\theta}$.
	\hfill{}
\end{proof}

\begin{prop}
	\label{claim:dec}
	Given any instance of \tumpg,
	\begin{align}
		S_{n-k}= S_{n-\theta} \cup S'_{n-k},\ \forall 1\leq k\leq n,\, 1\leq \theta \leq k
	\end{align}
	where $S_{n-k}$ is the set of nodes retained (corresponding to the
	upgraded base-stations) by applying \decg($k$) on $V$, and $S_{n-\theta}$,
	$S'_{n-k}$ are the sets of nodes retained by two consecutive operations
	of \decg($\theta$) on $V$, and \decg($k-\theta$) on $V \setminus
	S_{n-\theta}$ respectively.
\end{prop}

\begin{proof}
	Observe that the algorithm always prunes the nodes in $V$ in a consistent
	order. This follows from the fact that the algorithm breaks all ties between
	the nodes deterministically (if needed, using the index of nodes).
	Therefore, without loss of generality, assume that \decg($0$) applied on the
	node set $V$ prunes the nodes in the order $v_1,\dots,v_n$.  Hence,
	applying \decg($k$) on $V$ would prune the nodes $v_1,v_2,\dots,v_{n-k}$
	and return $S_{n-k}=\{v_{n-k+1},\dots,v_n\}$.  For $\theta \leq k$,
	$S_{n-\theta}=\{v_{n-\theta+1},$ $\dots,v_n\} \subseteq S_{n-k}$.  Applying
	\decg($k-\theta$) on the set $V \setminus S_{n-\theta} =
	\{v_1,\dots,v_{n-\theta}\}$ would prune the nodes $v_1,\dots,v_{n-k}$ and
	return $S'_{n-k}=\{v_{n-k+1},\dots,v_{n-\theta}\}$.  Hence,
	$S_{n-k}=S_{n-\theta} \cup S^\prime_{n-k}$.
	\hfill{}
\end{proof}
}
{
\comment{
\begin{prop}
Let $k=k_1+k_2+\dots+k_t$ such that $k_r >0$ for all $1 \le r \le t$, and $t \ge 2$. Then for any instance of \tumpg, \incg (respectively \decg) with budget parameter $k$ would upgrade the same set of base-stations, as successive applications of \incg (respectively \decg) with budget parameters $k_1,k_2,\dots,k_t$, respectively. 
\end{prop}
\textsc{Proof:} The proof is available at \cite{FullTUMP}.
}
}
\longv
{
\subsection*{Summary of the Algorithms}
A summary of analysis of the proposed greedy algorithms is presented in Table~\ref{tab:summary of analysis}. While \simg, and \incg have lower time complexities than \decg, but they do not offer any bounded approximation. \decg, on the other hand, offers a bounded approximation with the trade-off of higher running time. 
}
{}

\begin{table}[t]\scriptsize
		\caption{Summary of Analysis of Algorithms for \tumpg}
		\label{tab:summary of analysis}
	\begin{center}
		\begin{tabular}{|c||c|c|c|}
			\hline
		Algorithms	 & Time Complexity & Space Complexity & Approx. Bound\\
\hline	
\simg & $O(m\dmax +n\log k)$ & $O(m\dmax)$ &	0\\ 
			\hline
	\incg & $O(m\dmax^2 + kn)$ & $O(m\dmax)$ & 0 \\
			\hline
	\decg & $O(m\dmax^2 + (n-k)n)$ & $O(m\dmax)$ &  ${\binom{k}{\dmax}} /	{\binom{n}{\dmax}}$ \\				
			\hline			
		\end{tabular}
	\end{center}
\end{table}

\section{Evaluation Methodology}
\label{sec:methodology}

We show the efficacy of the algorithms using real as well as synthetic but realistic datasets. 
Our formulation considers two critical parameters whose values are best decided by the operator: the budget ($k$), and the threshold $\gamma$ (which signifies the desired satisfaction level on the trajectory that the operator is targeting). The input trajectories are generally chosen by the cellular operator based on the target subset of subscribers (e.g., based on focused micro-segments such as long-commuting 3G subscribers). The trajectories of subscribers are readily available to operators in Call Detail Records (CDRs) and deep-packet inspection logs~\cite{tdrs}.

\subsection{City and Network Generator (CING)}
\label{sec:cityModel}
\longv{
Real data about trajectories of subscribers is generally not publicly available. Hence, we designed \textit{City and Network Generator} (CING) tool to generate representative traces of population distribution, mobility and network topology.
We evaluate our protocols under such synthetic data and, also, on one real trace data set~\cite{realityMining}; we describe these data sets later.
CING models both the city population and network deployment by considering: (1) how users travel in a city, and (2) how base-stations (BS) are deployed. We now briefly explain these aspects.
}
{
Real data about trajectories of subscribers is generally not publicly available. Hence, we designed \textit{City and Network Generator} (CING) tool to generate representative traces of population distribution, mobility and network topology.
CING models both the city population and network deployment.
}

\noindent \textbf{City Model:}
\longv
{
\begin{figure}
\begin{center}
\includegraphics[width=4cm]{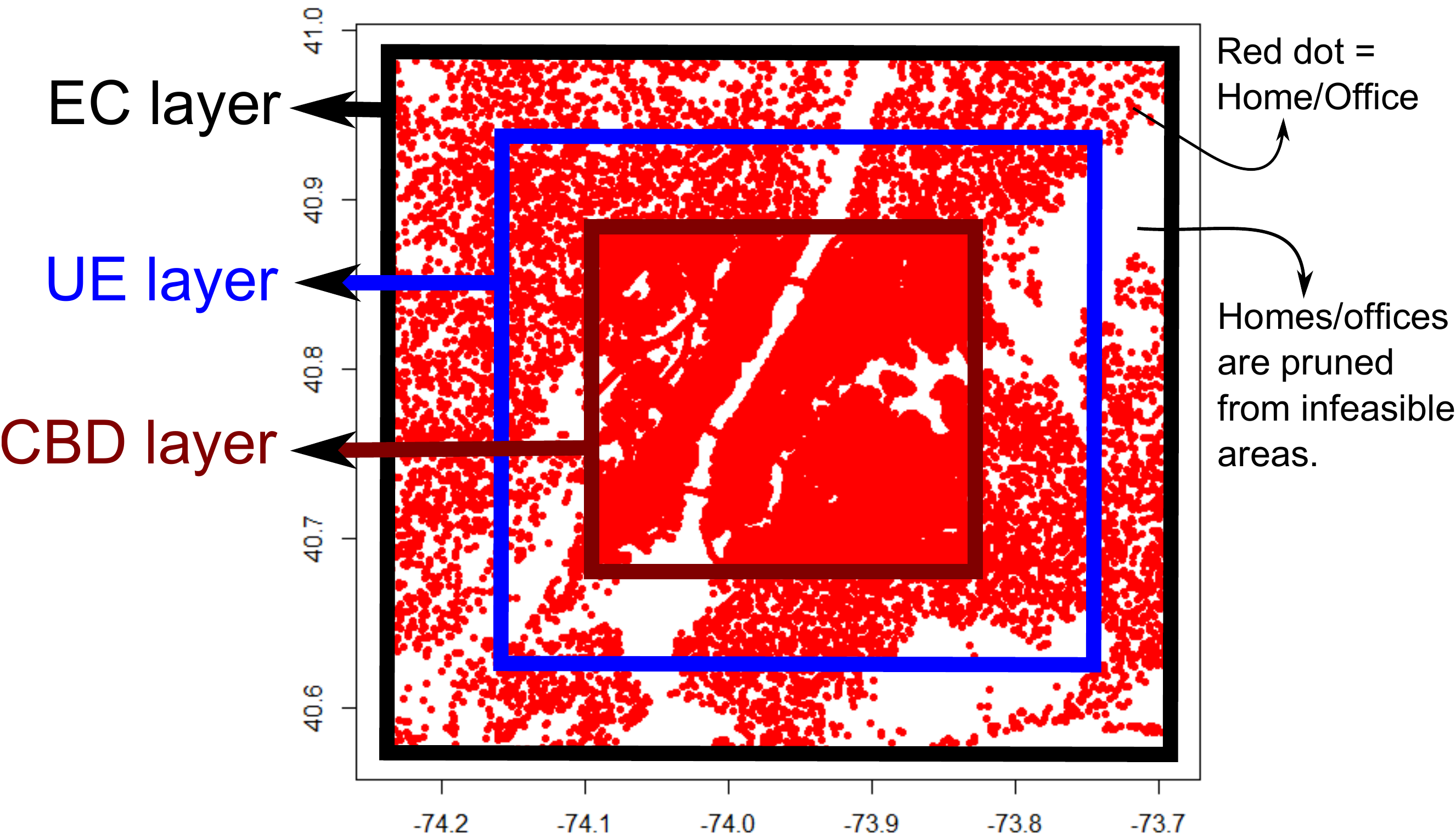}
\end{center}
\caption{CING City Model: Example of a virtual NYC}
\label{fig:cingCity}
\end{figure}

We model a virtual city based on the existing map-data and the observed spatial distribution of homes and offices.
We employ a concentric city model, where the city is divided into a few layers~\cite{beyondEdgeless}. The \textit{Central Business District} (CBD) constitutes the city-center, and usually has high office density. The \textit{Urban Envelopes} (UE) are around the CBD. \textit{Edge Cities} (EC) surrounds the main city's UE.

CING reads parameters such as dimensions and office/residential densities in various layers of the city, and creates synthetic user homes and offices. Lang \text{et. al.} have reported these information for thirteen popular US cities~\cite{beyondEdgeless}. However, only some regions of the city has been surveyed. Since we are interested in movement of subscribers in a city, directly using the partial city information leads to biased trajectories. Hence, we extrapolate the parameters reported to the entire city dimension.
City dimension is computed by latitude and longitude information the outer-most layer of the city that is surveyed (the EC layer). We then construct concentic rectangles of CBD, SD, UE and EC layers for each city as shown in Fig.~\ref{fig:cingCity}.

The areas of the layers are scaled proportionally to match the overall area. We distribute homes/offices in each layer according to the observed spatial densities~\cite{beyondEdgeless}. 
Using map data, we then prune the locations that are in inaccessible regions, such as locations within the rivers (see Fig.~\ref{fig:cingCity}).

Currently, we associate home and office as \textit{User hangouts}. A user hangout is represented by a spatial location and most frequent user's arrival and departure times. We artificially populate the home and office hangouts by observing weekday commute patterns of the users.; for example, by observing that the user leaves home at a random time from \unit{7}{AM} to \unit{11}{PM} to commute to office. If the spatio-temporal hangout information is available (e.g., as a result of mining real profile data, social network data or CDR data), the data can hence also be provided as an input to CING tool rather than generating virtual locations.
}
{
We employ a concentric city model, where the city is divided into a few layers~\cite{beyondEdgeless}. The \textit{Central Business District} (CBD) constitutes the city-center, and usually has high office density. The \textit{Urban Envelopes} (UE) are around the CBD. \textit{Edge Cities} (EC) surrounds the main city's UE. We distribute homes/offices in each layer according to the observed spatial densities~\cite{beyondEdgeless}. 
}
\noindent\textbf{Network Model:}
CING generates base-stations (BS) such that the number of BS in a region is proportional to the number of homes and offices. Our network consists of 82\% 2G BS, based on the deployment statistics of a major Indian cellular operator~\cite{airtel}. We mark 20\% of the BS as congested. Based on our measurement observations (Section~\ref{sec:measurement}), we randomly choose per-user throughput within $[20,80]$, $[50,150]$, $[20,400]$ and $[300,2000]$ kbps for congested 2G, non-congested 2G, congested 3G and non-congested 3G BS respectively.

\noindent\textbf{Trajectory Model:}

The trajectory properties such as the path taken from origin to destination affects the user's experience. CING derives the road-network graph from the OSM map-data~\cite{osm}. It computes the trajectories of the user by simulating user movement on road between home to office (using the shortest path algorithm). For each road trajectory thus computed, we compute the sequence of base-stations by assuming a hand-off policy where a user is connected to the nearest base-station at any location. The time connected to a base-station based on the road-speeds and BS coverage. The throughput of the user is dependent upon the per-user throughput of the BS and the duration of association.

\subsection{Data sets}
\label{sec:data sets}
We evaluate the performance of the proposed algorithms using both real and synthetic data.

\noindent\emph{a) Data set from real traces:} The Reality Mining (RM) data set lists the base-station handoffs of more than $100$ users~\cite{realityMining}. Most of the users belong to MIT and, hence, the set is biased. 

\longv
{
We use this data set to study optimizing trajectories of targeted subscribers with similar spatial hangouts and interests. We extract the sequence of base-station IDs that a user connects. We break the sequence of base-stations into trajectories where the user might have moved; we break a trajectory when there is no change in base-station for more than $30$ minutes. We prune the data to include mobile trajectories of users which have more than 10 base-stations. We also remove trajectories with sequence of base-stations showing a ping-pong handoff pattern~\cite{Juang2005}. The data set has 3819 trajectories and 17975 base-stations.

We first extract the trajectories by assuming that the a trajectory starts if the user is previously static for more than 30 minutes. The data also contains noisy trajectories, primariy because of: (1) base-station IDs which are out of the city are recorded, and (2) handoffs that often cause \textit{ping-pong} effect where the device swaps between a set of base-stations even when the user is not moving~\cite{Juang2005}. Such outlier base-stations and trajectories are pruned; we remove 1000 base-stations that have the least number of trajectories passing through them, and only consider trajectories that have a length of greater than 10.


We finally get $3819$ trajectories and $17975$ base-stations. Note that there is no base-station location information in this data; however, ID suffices for our evaluation.
}
{}

\begin{figure}
\begin{center}
\includegraphics[width=3.95cm]{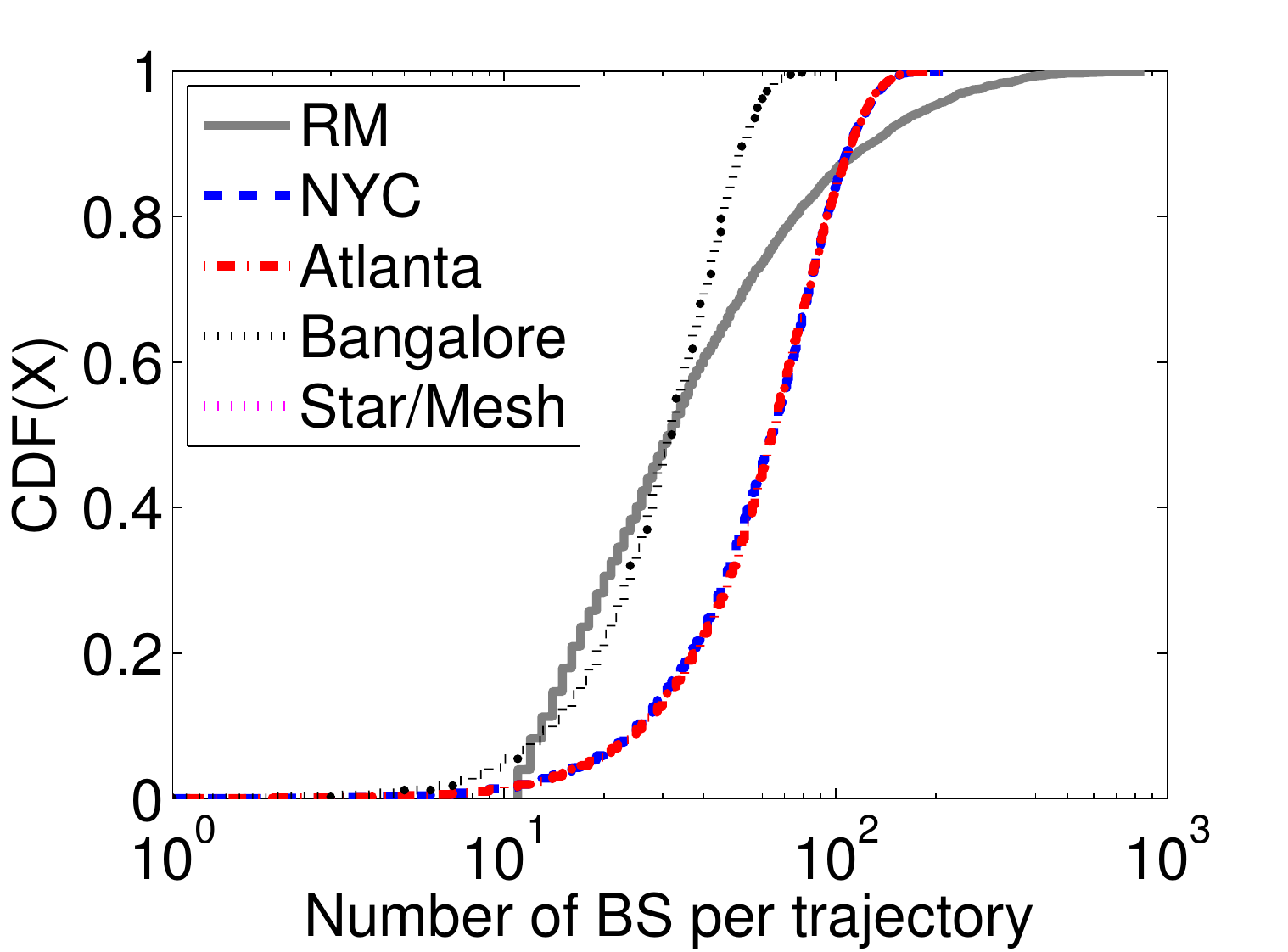}
\includegraphics[width=3.95cm]{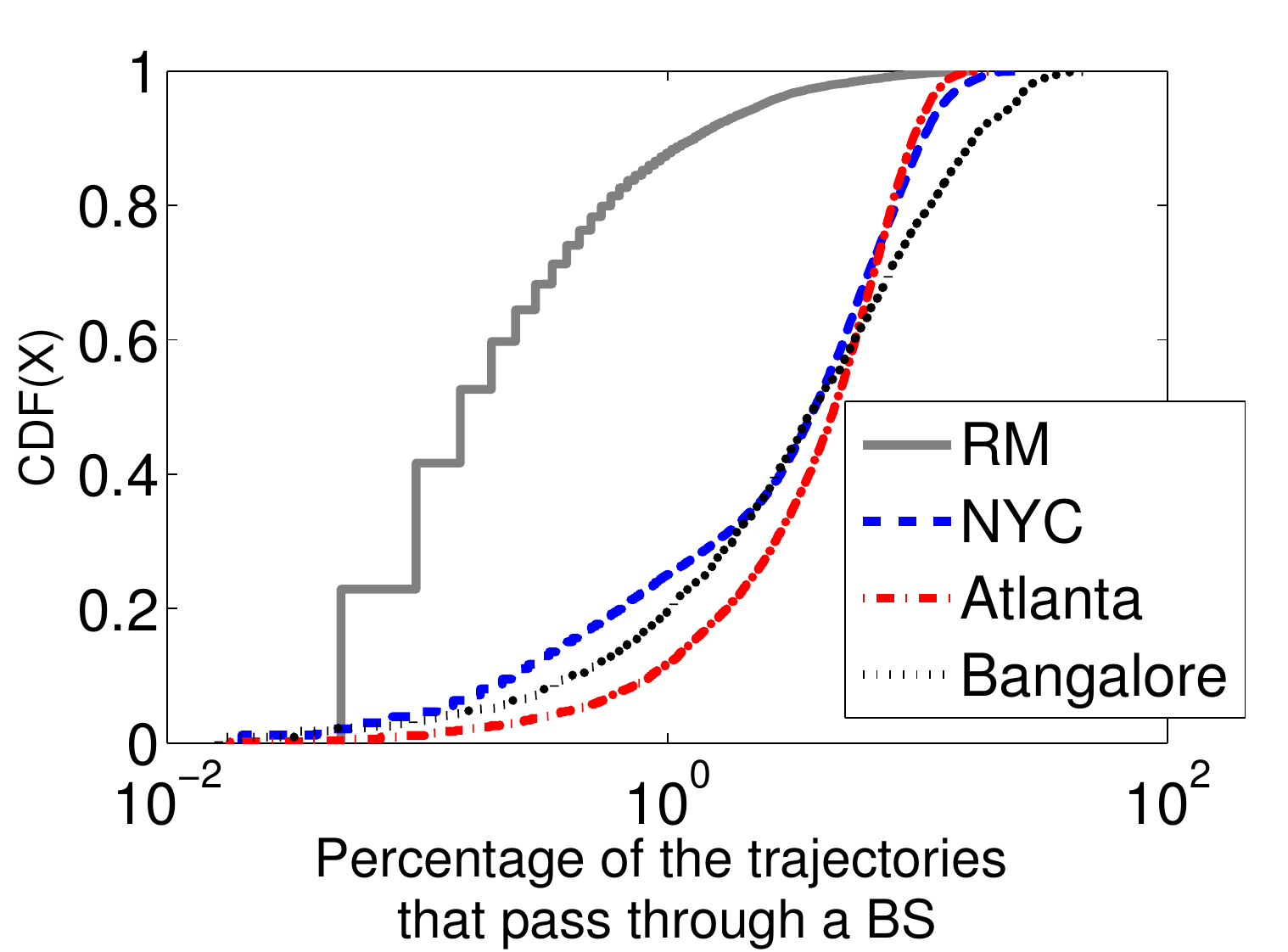}
\end{center}
\caption{Distribution of trajectory length and degree of a base-station: The degree of the base-station is normalized to the number of trajectories present in the scenario.}
\label{fig:dataDist}
\end{figure}
RM has 3,819 trajectories and 17,975 base-stations. Fig.~\ref{fig:dataDist} shows
the CDF of the trajectory length and the degree of a base-station (normalized to
the percentage of trajectories that pass through a base-station). The data shows
a highly skewed distribution, with many base-stations having very low degree
(when compared to the city movement in pure synthetic city data). We conjecture
that this is due to the biased set of users; many trajectories often visit a very few base-stations near MIT.

\begin{table*}\scriptsize
\caption{Synthetic data sets: We also evaluate Mesh topology and real traces from Reality Mining (RM) data set.}
\label{tab:data}
\begin{tabular}{|c|c|c|c|c|c|c|}
\hline
Star&\multicolumn{2}{|c|}{Bangalore}&\multicolumn{2}{|c|}{Atlanta}&\multicolumn{2}{|c|}{New York (NYC)}\\
\hline
\parbox[c]{2cm}{\includegraphics[width=1.8cm]{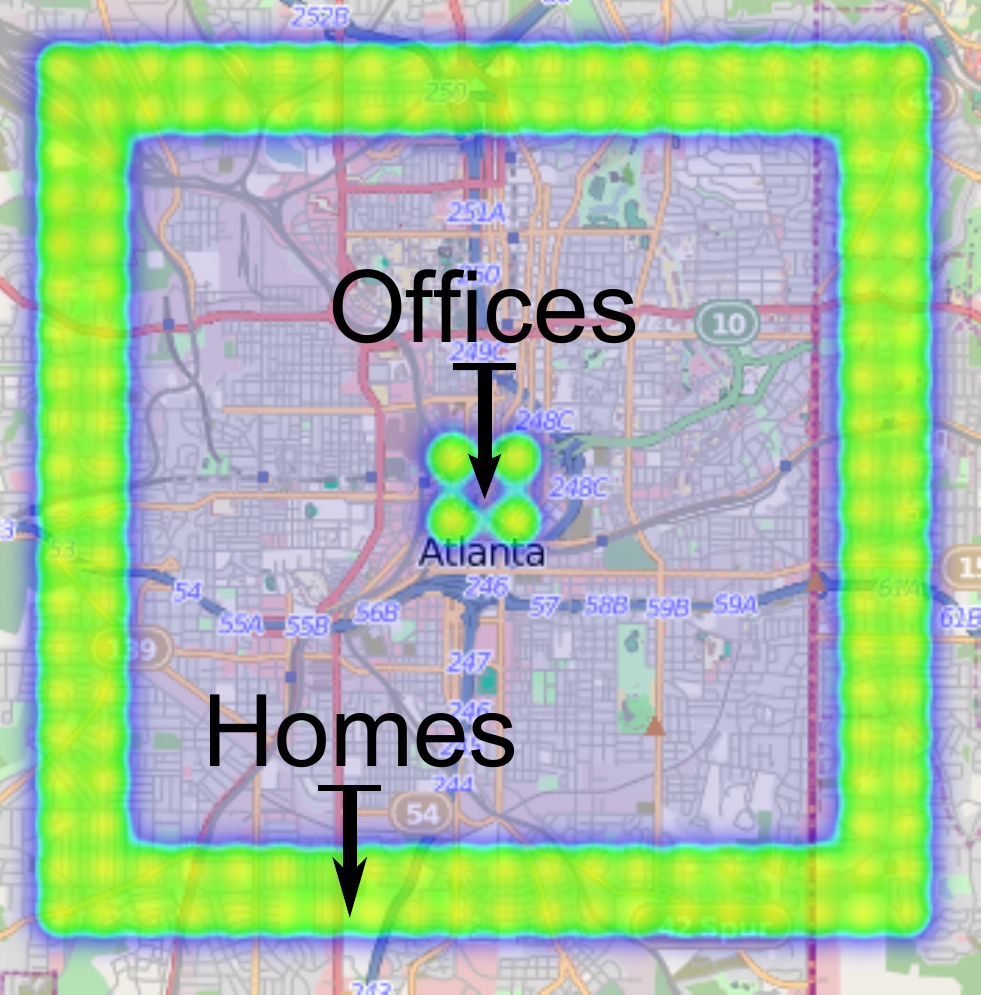}}&\parbox[c]{2cm}{\includegraphics[width=1.8cm]{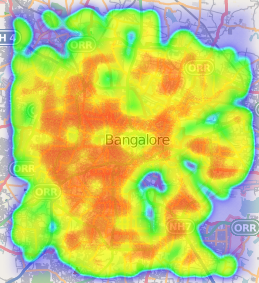}}&\parbox[c]{2cm}{\includegraphics[width=1.8cm]{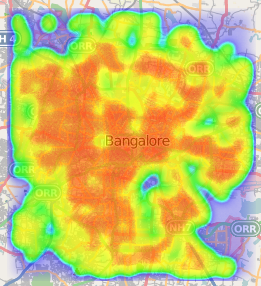}}&\parbox[c]{2cm}{\includegraphics[width=2cm]{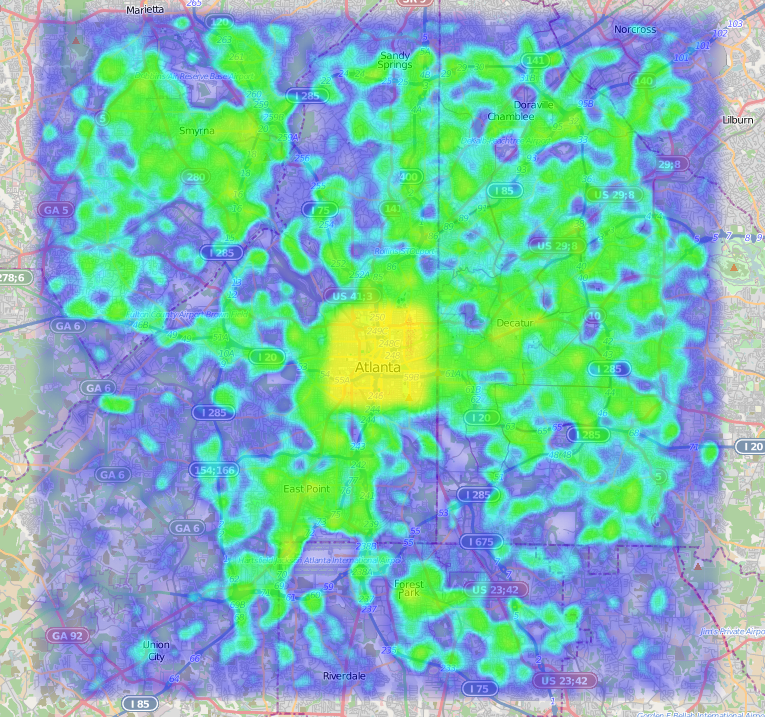}}&\parbox[c]{2cm}{\includegraphics[width=2cm]{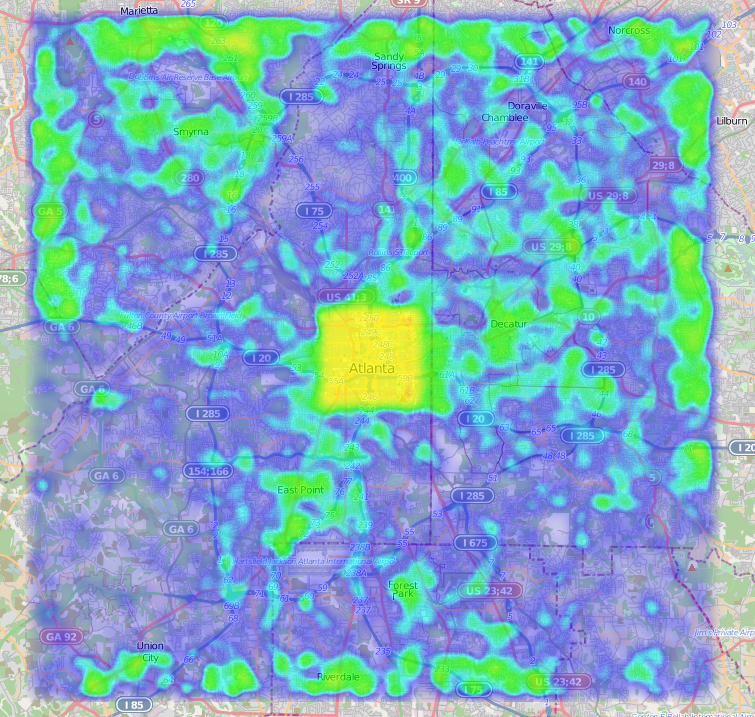}}&\parbox[c]{2cm}{\includegraphics[width=1.8cm]{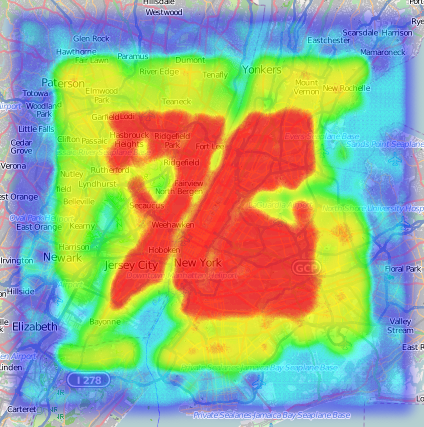}}&\parbox[c]{2cm}{\includegraphics[width=1.8cm]{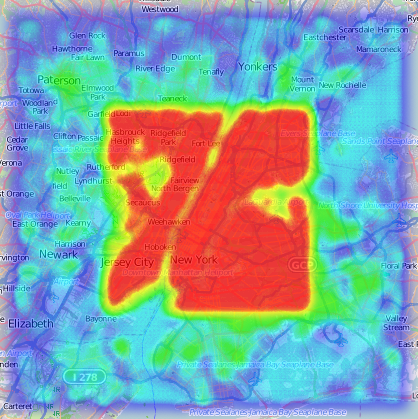}}\\
Home and Office&Home&Office&Home&Office&Home&Office\\
\unit{$10\times14$}{$\text{km}^2$}&\multicolumn{2}{|c|}{\unit{$18\times18$}{$\text{km}^2$}}&\multicolumn{2}{|c|}{\unit{$42\times61$}{$\text{km}^2$}}&\multicolumn{2}{|c|}{\unit{$46\times64$}{$\text{km}^2$}}\\
\unit{$5$k}{traj}; \unit{$430$}{BS}&\multicolumn{2}{|c|}{\unit{$25$k}{traj}; \unit{$1,894$}{BS}}&\multicolumn{2}{|c|}{\unit{$50$k}{traj}; \unit{$11213$}{BS}}&\multicolumn{2}{|c|}{\unit{$50$k}{traj}; \unit{$13860$}{BS}}\\
\hline
\end{tabular}
\end{table*}

\noindent \emph{b) Synthetic data sets:} We generate three classes of datasets, as summarized in Table~\ref{tab:data}.
\\\noindent 1. \textbf{Star} and \textbf{Mesh}: These topologies have artificially generated distribution to highlight upgrades in extreme cases of population distribution. Star has a dense CBD with offices, and users commute from their home in a thin UE layer. Mesh indicates a large city where people are equally likely to move in any direction. Here, trajectories are randomly assigned to base-stations with trajectory length distribution similar to Star.
\\\noindent 2. \textbf{New York City (NYC)} and \textbf{Atlanta}: We generate two representative large US cities of differing population distributions~\cite{beyondEdgeless}. 
NYC has a large CBD with dense offices, while Atlanta's CBD is small and relatively sparse. Hence, NYC dataset consists of more concentric trajectories, and Atlanta's trajectories are spread out.
\\\noindent 3. \textbf{Bangalore}: We simulate an Indian city, where the population distribution is different than USA. Offices are concentrated in business areas, and homes are spread out across city.

Fig.~\ref{fig:dataDist} shows the trajectory length and percent of trajectories incident on base-stations in different data sets. When compared to the RM data set, the synthetic set has higher degree of trajectories incident on each base-station since we consider people moving from all parts of the city towards their respective offices.

\section{Results}
\label{sec:results}
We measure the effectiveness of the algorithms by computing the percentage of trajectories that are \textit{$\gamma$-bottleneck-free}, which is an indication of the percentage of mobile users who will be satisfied by the upgrade. We first analyze small-scale regions in the city to understand the 
effect of model parameters, and then we show the improvements in large-scale city-wide scenarios.

\subsection{Comparison with the Optimal}
\label{sec:optComp}

\begin{figure}
\begin{center}
\subfigure[Performance of \decg is close to optimal
~\label{fig:opt_perf}]{\includegraphics[width=4cm]
{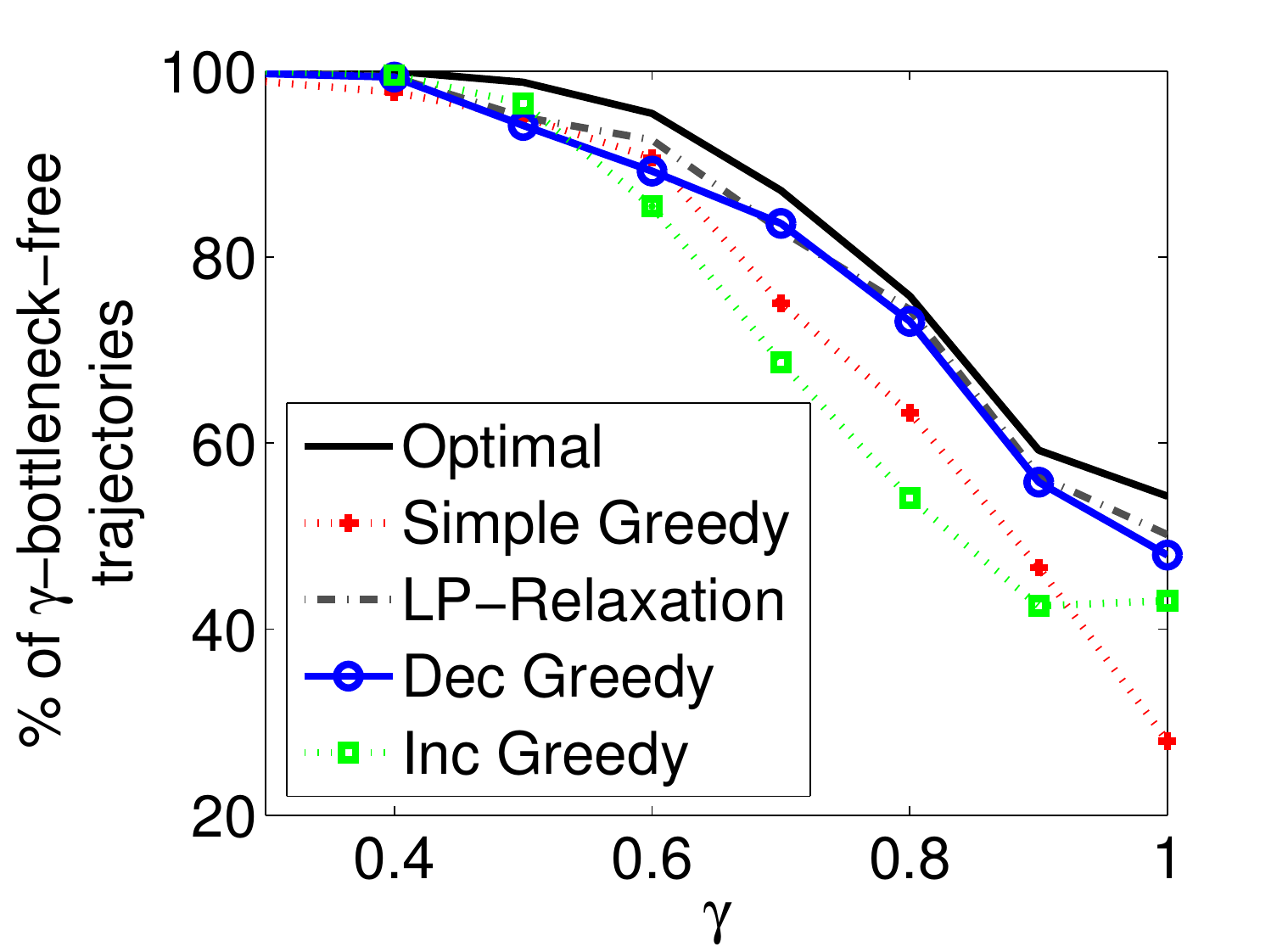}}\quad
\subfigure[Time comparison: Our algorithms are 3-4 orders faster
~\label{fig:opt_time}]{\includegraphics[width=4cm]
{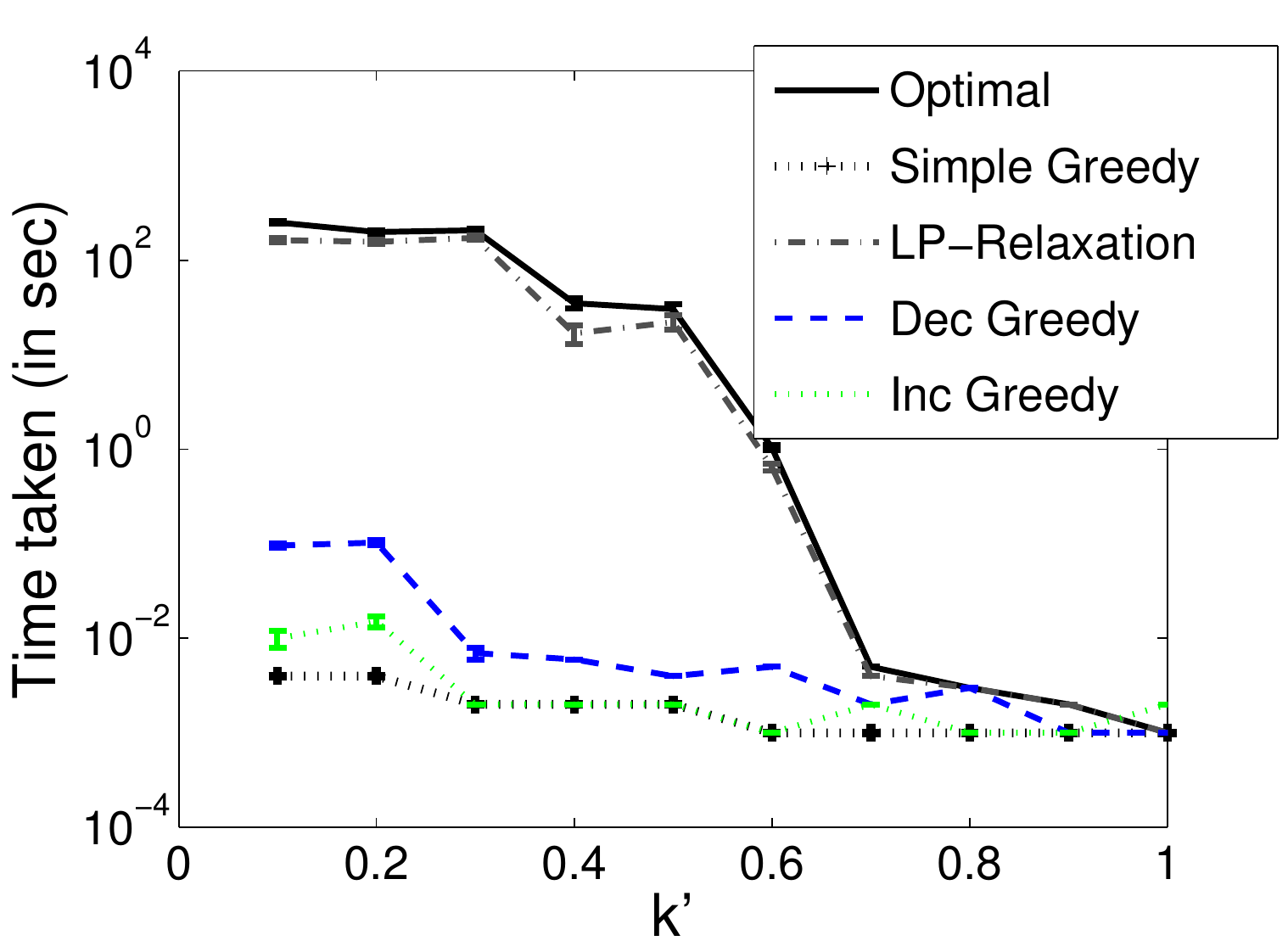}}
\end{center}
\caption{Comparison of different schemes w.r.t. Optimal}
\label{fig:optcomp}
\end{figure}
We extracted the trajectories in a miniature \unit{2x2}{km$^2$} area from Bangalore 
dataset.
Since the problem is NP-hard, the optimal solution does not scale 
up for larger city networks.
Therefore,
we compared the performance of our algorithms against the optimal solution 
on this small dataset.
This region had 1405 trajectories and 30 base-stations, out of which we
chose to optimize 9 base-stations. Since each dataset has different
number of base-stations, we analyze the results with the fraction of BS to be 
upgraded, $k'=\frac{k}{n}$. In the above scenario, $k' = \frac{9}{30} = 0.3$.

We also compare with LP-Relaxation based approach~\cite{lprelax}, which was
discussed in Section~\ref{sec:algo}. 
Fig.~\ref{fig:opt_perf} compares the performance of the algorithms at different $\gamma$'s,
and Fig.~\ref{fig:opt_time} shows the run-time for $\gamma=1$. Our algorithms,
especially \decg, provide similar performances as optimal and LP-Relaxation, 
while running 3-4 orders of magnitude faster. 
\simg fails to provide good user experience at higher $\gamma$.

\incg is also closer to optimal; however the selfish choices of the 
algorithm during initial iterations results in a lower performance. We explain the 
reasons in more detail later.

\subsection{Effect of budget allocation, QoS and trajectory lengths}
We choose a \unit{5x5}{km$^2$} area in Bangalore to demonstrate the detailed effects of basic parameters; a larger scale city analysis is later discussed in Section~\ref{sec:largescale}. This area has 23,000 trajectories with 107 BS. There are 4240, 5420, 8470 and 5430 trajectories with lengths between $[1,5]$,$[6,10]$, $[11,15]$ and $[16,20]$, respectively. 
\begin{figure}
\begin{center}
\subfigure[Varying budget (k'); $\gamma=1$.
~\label{fig:trajLth_k}]{\includegraphics[width=4cm]
{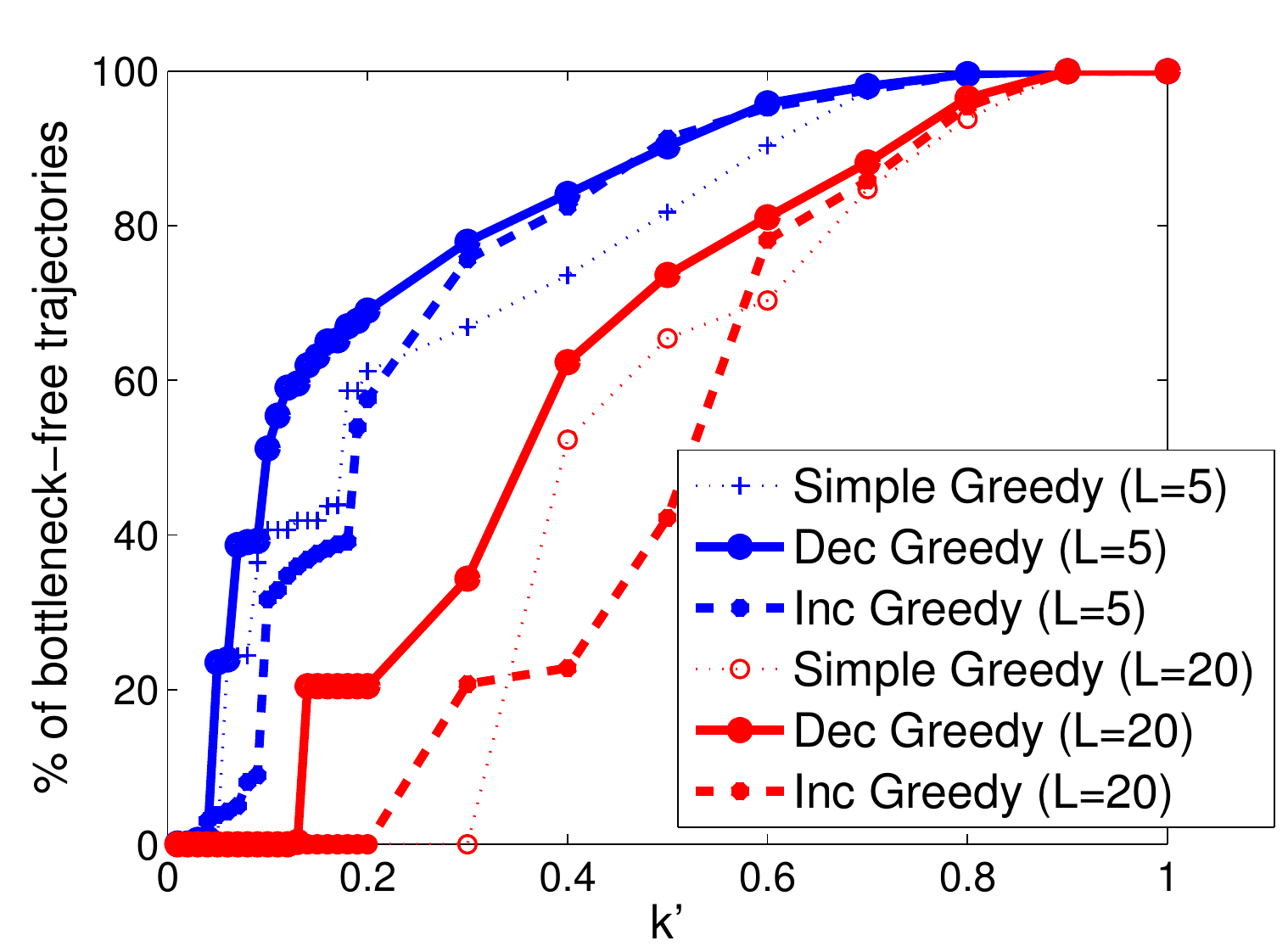}}\quad
\subfigure[Relaxing QoS ($\gamma$); $k'=0.2$
~\label{fig:trajLth_gamma}]{\includegraphics[width=4cm]
{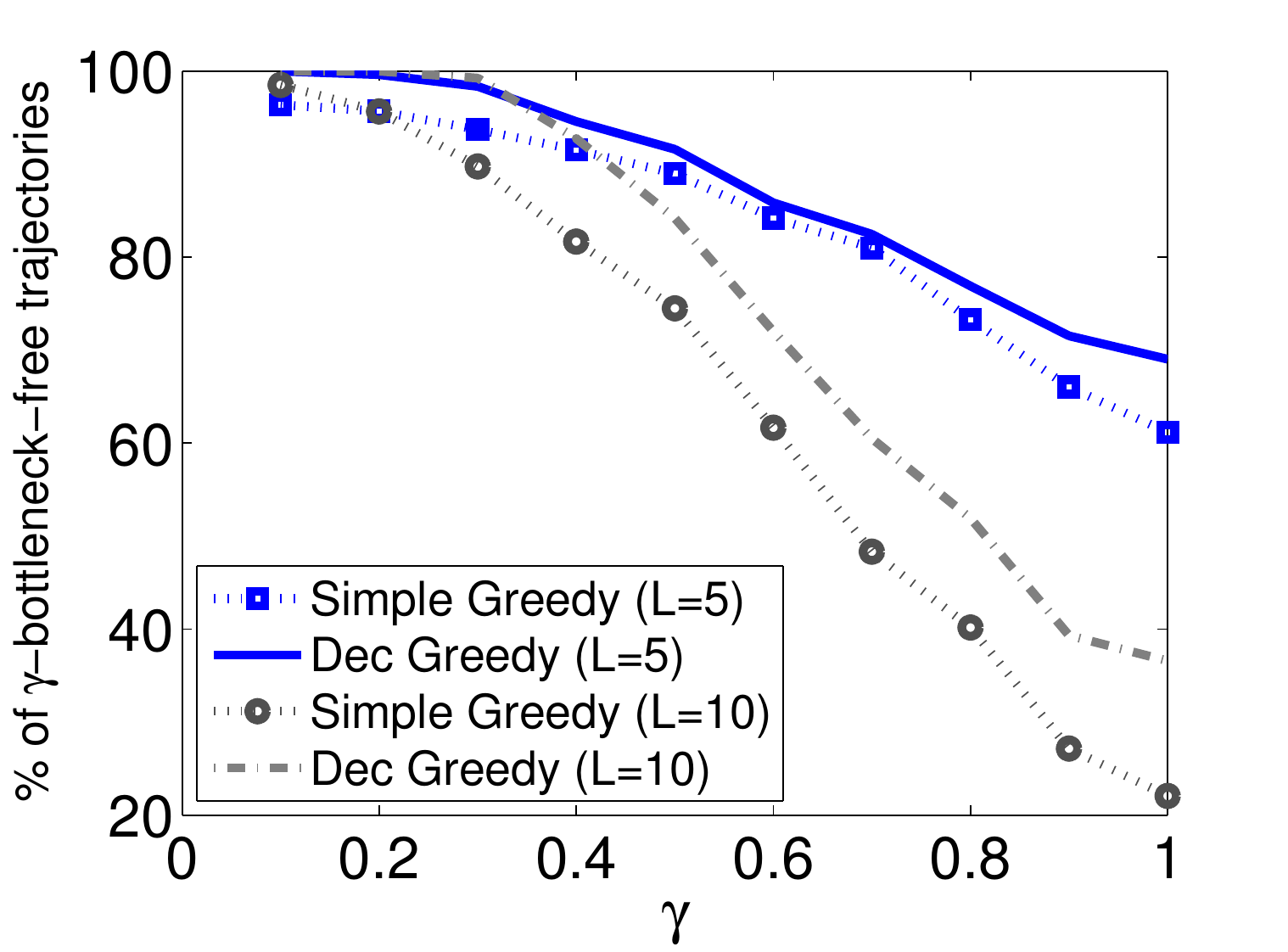}}
\end{center}
\caption{Effect of trajectory length (L), $k'$ and $\gamma$: Our algorithms are well-suited for network upgrades when stricter QoS has to be delivered over longer trajectories with low budgets}
\label{fig:trajLths}
\end{figure}
\\\noindent\textbf{1. Effect of budget allocation: }
Operators incrementally allocate small budgets for upgrades. For example, $k' \le 0.2$ is a representative annual budget of a major Telco~\cite{airtel}. Fig.~\ref{fig:trajLth_k} shows the effect of altering the budget ($k'$) when the operator chooses to ensure strict QoS to the mobile user ($\gamma=1$). \decg consistently outperforms \simg across different budgets. Note that \decg can be repeatedly applied in increments as more budget becomes available. Hence, the operator can consistently provide better performance than \simg as the network evolves. 

\incg yields better results than \simg at low $k'$. However, its performance is worse than \simg at intermediate $k'$-values. This happens as, at each iteration, \incg chooses to upgrade the BS that provides immediate benefits; the benefit lies in converting a bottleneck trajectory to $\gamma$-bottleneck-free. This short-term focus on immediate satisfaction biases \incg to miss out on choosing BSes that are beneficial to a larger number of trajectories which could have been possibly identified in later iterations. 
\\\noindent\textbf{2. Effect of trajectory lengths: }
Fig.~\ref{fig:trajLth_k} shows the effect of optimizing experiences along short and long trajectories. If the operator chooses to upgrade short trajectories, \simg approach provides comparable performance as \decg. However, if the long-commuting trajectories are to be optimized, \decg approach provides substantial gains, especially for small budgets. 

The \incg algorithm starts providing better gains for lower $k'$ values than \simg. However, \decg consistently outperforms the other two algorithms across different trajectory lengths. 
\\\noindent\textbf{3. Strictness of QoS: }
Fig.~\ref{fig:trajLth_gamma} shows the effect of relaxing the strictness of the QoS provided to the mobile user (reducing the value of $\gamma$).  For smaller 
$\gamma$ or for shorter trajectories, \simg scheme is comparable to \decg; there are a few highly visited BS on the trajectory, which are candidates for optimizing in \simg approach, which also renders the trajectory $\gamma$-bottleneck-free. However, when stricter user experience is needed over longer trajectories, \decg provides significant benefits over \simg.

\subsection{City-scale evaluation}
\label{sec:largescale}
\begin{figure}
\begin{center}
\subfigure[Reality Mining
~\label{fig:rm_k}]{\includegraphics[width=4cm]
{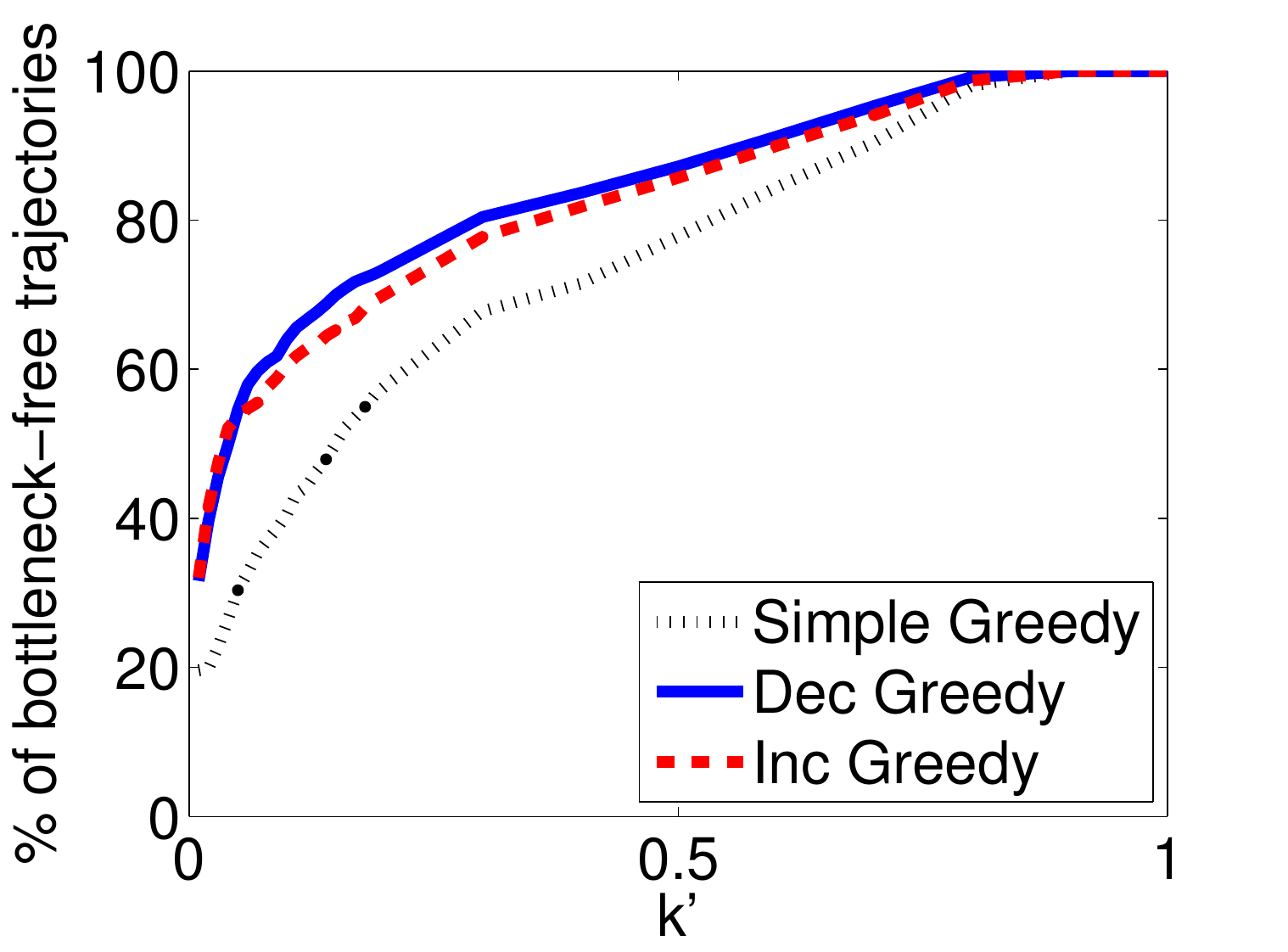}}\quad
\subfigure[Atlanta
~\label{fig:atlanta_k}]{\includegraphics[width=4cm]
{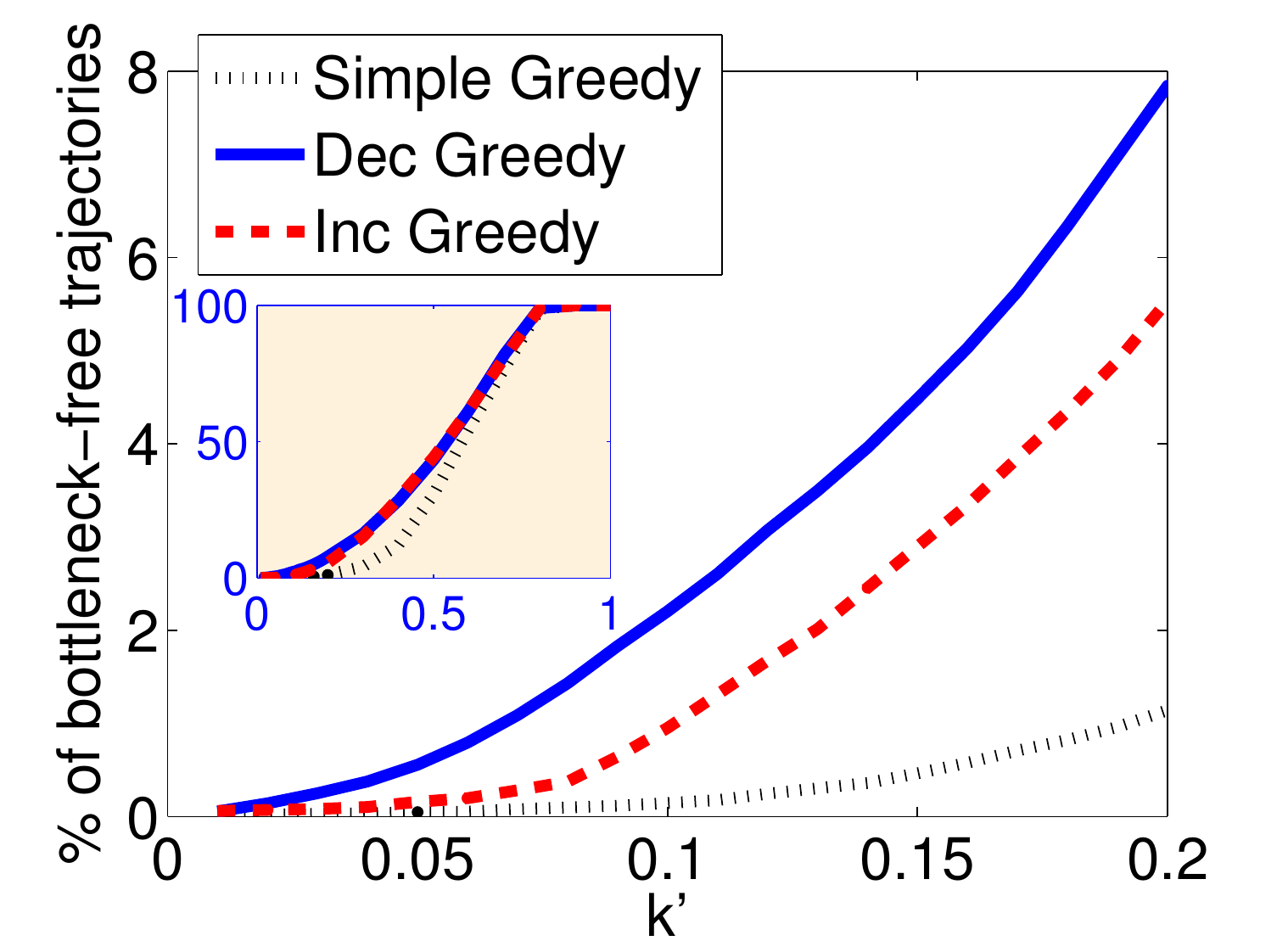}}
\end{center}
\caption{Performance of TUMP($\gamma=1$) in different scenarios: \decg and \incg consistently outperform \simg. With an upgrade budget $k'=0.2$, our algorithms perform 3-8 times better in improving the user quality of service on trajectories in different city geographies when compared to greedy location-based base-station upgrades.}
\label{fig:topoCities}
\end{figure}
Fig.~\ref{fig:topoCities} shows the performance of algorithms in different datasets. In the real RM scenario, \incg and \decg algorithms consistently provide more than 10-20\% improvement when the budget is medium to low. In this biased distribution, few BS (near MIT) are frequently visited by many trajectories. While \simg optimizes frequently accessed BS, it fails to optimize the  QoE of frequently accessed trajectories.

In the simulated datasets, \decg and \incg constantly outperform \simg. However, the improvement for low $k'$ is not as pronounced as in RM dataset since trajectories are spread across the entire city, and hence selecting relevant BS is hard. We also highlight the regions of $k' \le 0.2$ in Fig.~\ref{fig:topoCities}.

\begin{figure}
\begin{center}
\subfigure[Star-like \textit{vs.} Mesh-like~\label{fig:starMesh}]{\includegraphics[width=4cm]
{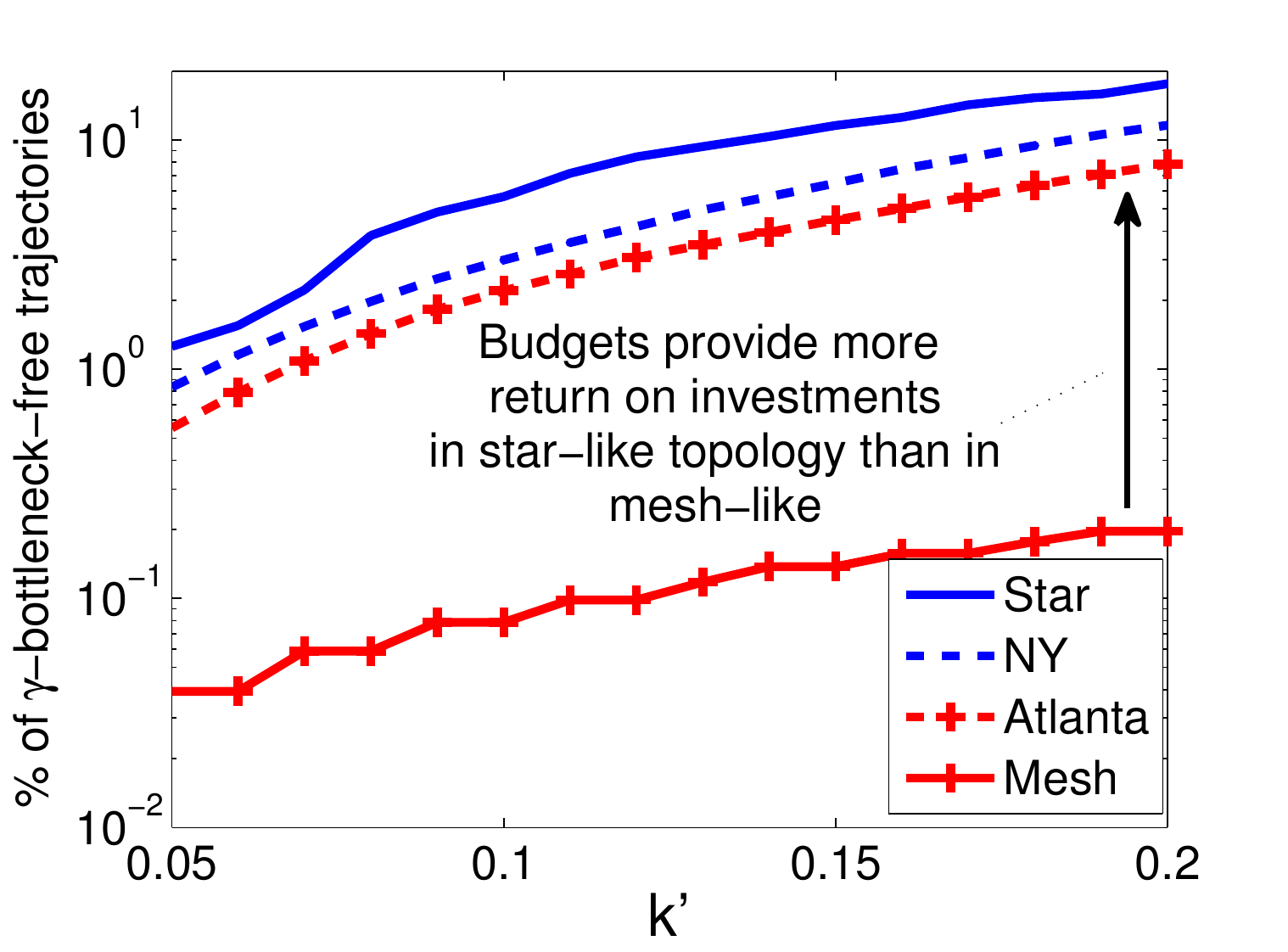}}
\subfigure[Runtime on Atlanta data-set
~\label{fig:atlantaTime}]{\includegraphics[width=4cm]
{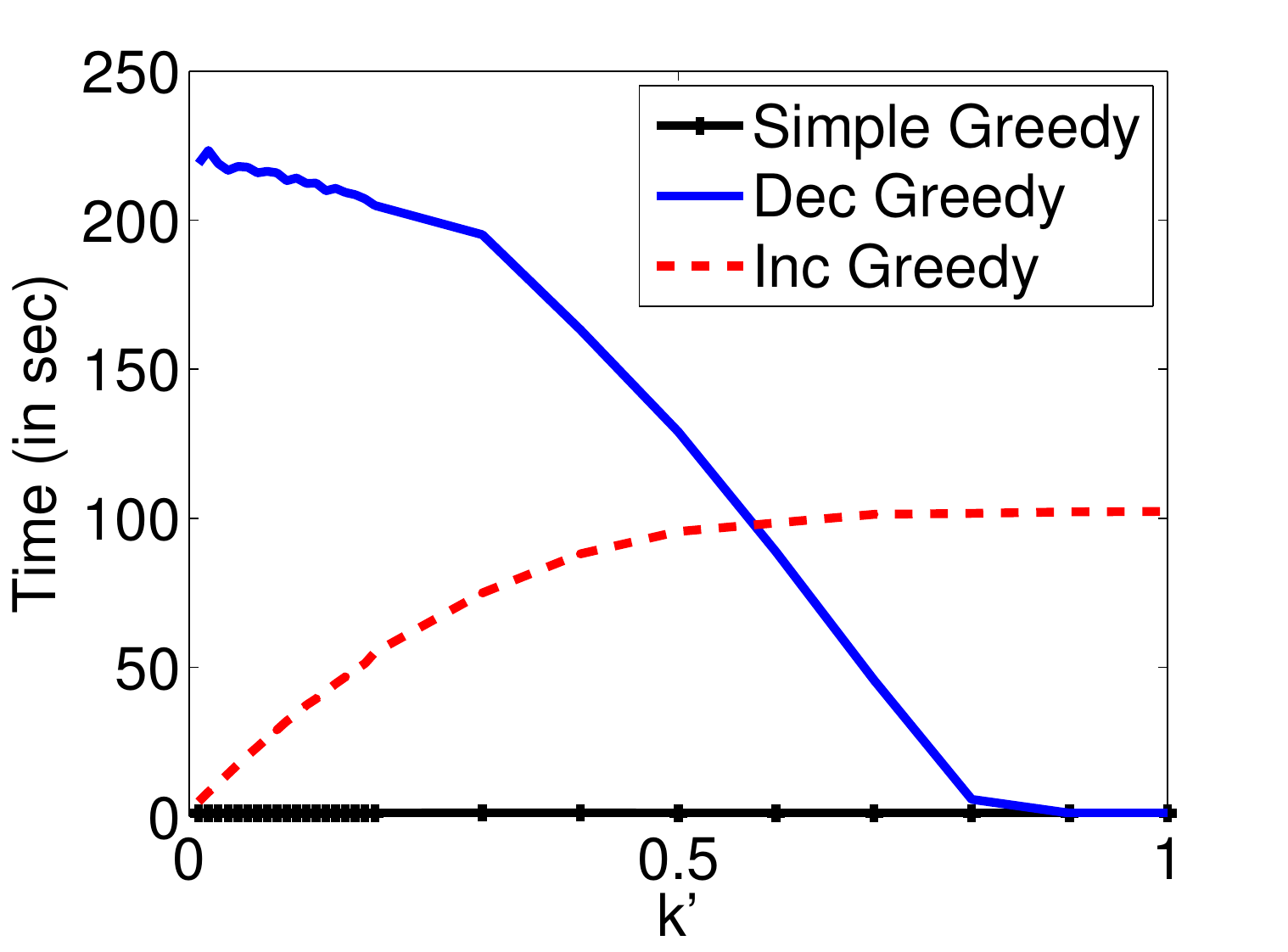}}
\end{center}
\caption{Impact of population and runtime: Star topologies provide more marginal returns for trajectory optimization. The runtime of \incg increases as $k'$ increases. However, \decg has a reverse trend.}
\label{fig:topoArtificial}
\end{figure}
\noindent\textbf{1. Implications of Population Distribution: }
We observed that NYC consistently provides higher marginal gains than Atlanta. This is because NYC has a high office density in the center, and trajectories form a star-like topology. Upgrading a few BS at the center, hence, benefits a large number of trajectories. In contrast, Atlanta's trajectories do not form a pronounced star-topology.

Fig.~\ref{fig:starMesh} compares the performance of \decg in: (1) artificial Star and Mesh topologies, and (2) realistic star-like NYC and mesh-like Atlanta,  to highlight the impact of population distribution.
Star-like topologies continuously provide significant improvement on small increases in budget. Unlike Star, Mesh requires a large budget ($k$) to provide good performance. For example, for a budget $k=0.2$, Star data-set provides two orders better improvement than Mesh.
Hence, while rolling out new upgrades in cities, the operator has to be cognizant that cost of providing better user experience in a star-like city is lower than in a mesh-like city.

\comment{
\paragraph{Impact of $k'$}
Fig.~\ref{fig:topoCities} shows that \decg and \incg provide better results than \simg in one-shot upgrades for most values of $k'$. In addition, as we show in Proposition~\ref{claim:inc}, our algorithms can be repeatedly applied in increments -- as the network gradually evolves -- to consistently provide better performance than \simg.

\paragraph{Impact of $\gamma$}
\begin{figure}
\subfigure[As $k'$ varies in Bangalore
~\label{fig:nycGamma}]{
\includegraphics[width=3.8cm]{Figures/blore_algs_difft_gamma.pdf}}
\subfigure[\decg in data-sets
~\label{fig:gammaDS}]{
\includegraphics[width=3.8cm]{Figures/nycAtlRM_algs_gammas.pdf}}
\caption{Percent of satisfied trajectories as $\gamma$ varies}
\label{fig:nycAtlantaGamma}
\end{figure}
Fig.~\ref{fig:nycAtlantaGamma} shows that as $\gamma$ decreases: (1) the number of bottleneck-free trajectories increase since $(1-\gamma)$ of the BS can be unoptimized for a $\gamma$-satisfied trajectory, and (2) the relative performance improvement of \decg and \incg algorithms observed over \simg decreases.
}

\noindent\textbf{2. Runtime comparison:}
Fig.~\ref{fig:atlantaTime} shows the runtime of different algorithms on the Atlanta data-set for $\gamma$=1. \decg and \incg are slower than \simg since both these algorithms are iterative. The running time of \incg increases as $k$ increases, whereas \decg has a reverse trend.

\section{Related Work}
\label{sec:related}


To the best of our knowledge, there is no literature that incorporates mobility 
pattern of the users for base-station upgrades. 
The related areas can be broadly classified into \emph{handoff protocols} and 
\emph{radio frequency (RF) planning}.  
\\\noindent(1) \textit{Handoff protocols:} These protocols improve 
connectivity by associating a device with a better base-station~\cite{Sgora2009}.  
They fail to improve experience if the base-stations are 
inherently limited in resources (e.g., 2G base-stations). 

Other techniques, such as Cell breathing, dynamically 
alter the coverage area based on the load~\cite{Sang2004}. While this provides 
high-capacity to static users, it does not explicitly address the problem of 
eliminating bottleneck across trajectories.

In contrast to the 
above protocol based approaches, we focus on the macro-level problem of 
base-station network planning.
\\\noindent(2) \textit{RF planning:} RF planning techniques optimize the 
transmission power, frequency, load and location of the base-station for 
providing better coverage and/or 
capacity~\cite{Song2010,Peng2011,Amaldi2008,Lifjens2001}.
RF planning configures base-stations to ensure adequate signal-to-noise ratio (SNR)
at the devices, considering the interference from neighboring 
cells~\cite{Glasser2005,Galota2001}.  Some radio planning also considers frequency allocation
and cell coverage optimization to cater to mobile users~\cite{Lifjens2001}.
Most of the RF planning tools, such as Atoll~\cite{atoll}, utilize drive 
tests and carrier wave (CW) measurements to recognize the 
bottlenecks~\cite{Song2010}. Such active measurements inherently 
limit measuring actual user experiences;  they cannot scale to measuring the 
millions of subscriber trajectories. 
Other short-range technologies~\cite{Banerjee2008,smallcell,femtocell} are 
susceptible to frequent handoffs for mobile users. 

In summary, unlike the above studies, we explicitly consider macro-cell 
planning for providing better user-experience \textit{along} users' 
trajectories. We use user's call logs, instead of passive measurements, to 
quantify mobile user experience. 

\section{Conclusions}
\label{sec:concl}

Macro-cell planning for mobile users is a relatively unexplored problem in the 
evolution of heterogeneous cellular network architectures.
This paper addressed the problem of performing macro-cell base-station upgrades 
by accounting for the mobile user satisfaction along their trajectories.

We conducted a measurement-based experiment to show that mobile users suffer 
from degraded quality of experience on their everyday commute trajectories.
Based on our findings, we formulated a generic problem that plans macro-cell 
upgrades to optimize mobile user experience. Our formulation utilizes active 
logs to quantify user experience. In addition, our formulation enables the 
operators to plan the upgrades in a way that is cognizant of their business 
needs and constraints. This is done by allowing the operators to define 
key parameters such as the budget for upgrades, desired satisfaction 
level on the trajectory, and a micro-segment of the mobile users whose 
mobile experience has to be optimized (e.g., high-value 4G customers with long 
commutes).

We proved the NP-hardness of the problem, designed two approximation 
algorithms, and proved their approximation bounds. We designed a 
synthetic city and network trace generator, and showed the dependence 
of planning budget and algorithm effectiveness in various 
population distributions.
Our algorithms consistently enable the operator to 
achieve 3-8 times better user experience than simple location-based 
greedy heuristics for the same budget.
In future, we plan to address problems in ``continuous network planning'' 
systems that ensure quality of experience for mobile users. 

\bibliographystyle{IEEEtran}
\balance
\bibliography{papers}

\end{document}